\documentclass[final,5p,times,twocolumn]{elsarticle} 

\usepackage{comment}
\usepackage{amsmath, amsfonts, amssymb, amsthm}
\usepackage{mathtools, mathcomp}  % added mathtools, mathcomp
\mathtoolsset{showonlyrefs=true}  % assign equation numbers only to 
\usepackage{mleftright}  %
\usepackage{bm}  % 
\usepackage{siunitx}  % 
\usepackage{lmodern}  % 
\usepackage{xcolor}  % 
\usepackage{enumitem}  %
\usepackage{flushend}  % 

\theoremstyle{plain}   % 
\newtheorem{theorem}{Theorem} 
\newtheorem{lemma}{Lemma} 
\newtheorem{proposition}{Proposition} 

\theoremstyle{remark}  %
\newtheorem{remark}{Remark}

\usepackage{lipsum}
\usepackage{threeparttable}

\usepackage[]{hyperref}
\hypersetup{%
 setpagesize=false,%
 bookmarks=true,%
 bookmarksdepth=tocdepth,%
 bookmarksnumbered=true,%
 colorlinks=false,%
 pdftitle={},%
 pdfsubject={},%
 pdfauthor={},%
 pdfkeywords={}
 }
\journal{Asian Journal of Control}

\begin{document}

\setlength{\abovedisplayskip}{4pt} 
\setlength{\belowdisplayskip}{3pt}

\begin{frontmatter}

\title{Fractional-order controller tuning via minimization of integral of time-weighted absolute error without multiple closed-loop tests}

\author[KYUDAI]{Ansei Yonezawa}\corref{correspendence}\ead{ayonezawa[at]mech.kyushu-u.ac.jp} %% Author name
\author[HOKUDAI]{Heisei Yonezawa}
\author[TCU,ISUZU]{Shuichi Yahagi}
\author[HOKUDAI]{Itsuro Kajiwara}
\author[KYUDAI]{Shinya Kijimoto}

\affiliation[KYUDAI]{organization={Department of Mechanical Engineering, Kyushu University},
                     %addressline={}, 
                     city={Fukuoka},
                     %postcode={}, 
                     %state={},
                     country={Japan}}

\affiliation[HOKUDAI]{organization={Division of Mechanical and Aerospace Engineering, Hokkaido University},
                     %addressline={}, 
                     city={Sapporo},
                     %postcode={}, 
                     %state={},
                     country={Japan}}
\affiliation[TCU]{organization={Department of Mechanical Engineering, Tokyo City University},
                     %addressline={}, 
                     city={Setagaya},
                     %postcode={}, 
                     %state={},
                     country={Japan}}
\affiliation[ISUZU]{organization={6th Research Department, ISUZU Advanced Engineering Center Ltd.},
                     %addressline={}, 
                     city={Fujisawa},
                     %postcode={}, 
                     %state={},
                     country={Japan}}



%%             country={}}  %% Author affiliation

\tnotetext[]{This work was supported in part by JSPS KAKENHI Grant Numbers JP23K19084 and JP25K17561, and in part by the funding provided by ISUZU Advanced Engineering Center.}
\tnotetext[]{\textcolor{red}{
This is the peer reviewed version of the following article: \\
\textit{A. Yonezawa, H.Yonezawa, S. Yahagi, I. Kajiwara, and S. Kijimoto, Fractional-order controller tuning via minimization of integral of time-weighted absolute error without multiple closed-loop tests, Asian J. Control (2025),1--14, DOI 10.1002/asjc.3788}, \\
which has been published in final form at \url{https://doi.org/10.1002/asjc.3788}. This article may be used for non-commercial purposes in accordance with Wiley Terms and Conditions for Use of Self-Archived Versions. This article may not be enhanced, enriched or otherwise transformed into a derivative work, without express permission from Wiley or by statutory rights under applicable legislation. Copyright notices must not be removed, obscured or modified. The article must be linked to Wiley's version of record on Wiley Online Library and any embedding, framing or otherwise making available the article or pages thereof by third parties from platforms, services and websites other than Wiley Online Library must be prohibited.
}}
\cortext[correspendence]{Correspondins author.}

%% Abstract
\begin{abstract}
This study presents a non-iterative tuning technique for a linear fractional-order (FO) controller, based on the integral of the time-weighted absolute error (ITAE) criterion. Minimizing the ITAE is a traditional approach for tuning FO controllers. This technique reduces the over/undershoot and suppresses the steady-state error. In contrast to conventional approaches of ITAE-based controller tuning, the proposed approach does not require multiple closed-loop experiments or model-based simulations to evaluate the ITAE. The one-shot input/output data is collected from the controlled plant. A fictitious reference signal is defined on the basis of the collected input and output signal, which enables us to evaluate the closed-loop response provided by the arbitrary controller parameters. To avoid repeated experiments that are necessary in the conventional approach, we reformulate the ITAE minimization problem using the fictitious reference signal. The desired FO controller parameters minimizing the ITAE are obtained by solving the optimization problem that is based on the fictitious reference signal. The validity of the proposed approach is demonstrated by a numerical study. The avoidance of repeated experiments significantly reduces the development cost of linear FO controllers, thereby facilitating their practical application.
\end{abstract}

%% Keywords

\end{frontmatter}

%%% Section 1 %%%
\section{Introduction} \label{1_Introduction}

\subsection{Motivation} \label{1_1_Motivation}
Fractional-order (FO) calculus \cite{Podlubny1998}, a mathematical theory on non-integer order differentiation and integration, has attracted considerable attention in various fields of science and engineering \cite{Sun2018}. In the field of systems and control engineering, FO calculus is often employed for system modeling \cite{Cheng2023, Xu2025}, observer design \cite{Naifar2022, Jmal2020}, and controller synthesis \cite{Mseddi2024, Wang2023}. 
Stability analysis and stabilization of FO systems have been widely explored (e.g., the Barbalat-type lemma for the conformable FO derivative \cite{BenMakhlouf2023}, practical stabilization of a class of nonlinear generalized conformable FO systems \cite{Gassara2022}).
Various linear FO control techniques have been proposed and their fundamental properties and design methods have been explored, including linear FO active disturbance rejection controllers \cite{Wang2024, Wang2023_ISA}, tilt-integral-derivative controllers \cite{Gnaneshwar2024, Lu2023}, and FO proportional-integral-derivative (FO-PID) controllers \cite{Podlubny1999, Chen2022, Ghorbani2023}. Linear FO controllers can outperform traditional integer-order (IO) controllers in terms of robustness \cite{Monje2023}, control efficiency \cite{Marinangeli2018}, and control performance \cite{Lino2017}. The superiority of FO control can be leveraged for both IO and FO plants \cite{Chen2006}. Therefore, exploring linear FO control methods is important for improving the safety and performance of various automatic control systems in practice. In fact, one study has emphasized the potential advantage of FO-PID controllers in industrial control systems \cite{Tepljakov2021}. Here, it should be noted that proper parameter tuning of FO controllers is one of the most important, albeit challenging procedures for maximizing the benefits of FO control \cite{Tepljakov2021}.

\subsection{Related studies} \label{1_2_Related studies}
Compared with the tuning of IO controllers, the tuning of FO controllers is more challenging owing to the lack of a clear physical interpretation of the FO integral/derivative and the increase in the number of the tunable parameters. To overcome these problems, considerable effort has been devoted toward developing a tuning technique for linear FO controllers. For example, several analytical conditions for FO controller parameters have been derived for specific classes of the controlled plants in the frequency domain, including first-order plus time-delay (FOPTD) systems \cite{Chen2022_ISA, Wu2024}, single FO pole plus time-delay systems \cite{Yumuk2022}, and second-order integrating plants \cite{Shankaran2022}. If these methods are adopted for general controlled plants, the plant must be modeled to meet the specific plant class, inevitably resulting in modeling errors, which often cause significant performance degradation. Moreover, the plant modeling itself imposes a considerable burden on designers and increases the development costs.

The optimization-based tuning strategy in the time domain is an alternative approach for tuning FO controllers. This tuning scheme relies on the control result data instead of the analytical model of the plant, as the time-domain performance index can be directly evaluated from the control result data \cite{Das2011}. Therefore, this approach can be adopted for various control systems, and it overcomes the drawback of the analytical model-based approaches discussed above \cite{Das2011}. Considering the generality and practical usability of this approach, numerous optimization algorithms have been investigated for tuning FO controllers, including particle swarm optimization (PSO) and artificial bee colony algorithms \cite{Bingul2018}, the genetic algorithm \cite{Mondal2020}, the grey wolf optimizer \cite{Faraj2023}, Rao algorithm \cite{Paliwal2022}, and the teaching-learning optimizer \cite{Veerendar2023}. Several studies have developed novel metaheuristics for tuning FO controllers \cite{Hekimoglu2019, Izci2023}. A comprehensive review of the optimization-based FO controller tuning has been presented in \cite{Tepljakov2021}. The typical time-domain performance index includes the integral of the absolute error (IAE), integral of the time-weighted absolute error (ITAE), and integral of the squared error; the detailed characteristics of these criteria have been discussed in \cite{Das2011}. Among them, the ITAE is often employed as the time-domain criterion, as minimizing the ITAE reduces the over/undershoot and suppresses the steady-state error \cite{Veerendar2023}.

Despite the utility of FO controller tuning via optimization as described above, this approach poses a major challenge in the performance evaluation of the controller. The conventional time-domain optimization-based tuning scheme can be outlined as follows:  (1) conduct closed-loop control tests; (2) evaluate the performance criteria using the control result from step (1); (3) determine the next candidate controller based on the evaluation result and the optimization strategy; and (4) repeat steps (1)--(3) until some termination condition is satisfied. Note that the closed-loop control test in step (1) is executed via control experiments using a real-world control system or control simulations using the plant model. Repeating the control experiments in a real-world system is not only tedious but also dangerous in some situations. Model-based control simulations involve costly mathematical modeling of the controlled plant; moreover, the inevitable modeling error in the simulation model causes performance degradation. In other words, the need for the iterative closed-loop control test is a significant obstacle in optimization-based controller tuning. Consequently, a new paradigm in the optimization-based tuning of linear FO controllers is required to improve the utility of this tuning technique. Establishing a simple tuning scheme can facilitate the use of FO controllers in various automatic control systems.

\subsection{Contribution and novelty} \label{1_3_Contribution}
This study proposes a new parameter tuning technique for linear FO controllers. The controlled plants are discrete-time (DT) single-input single-output (SISO) linear time-invariant (LTI) systems. Although the proposed tuning approach relies on the ITAE minimization strategy, it needs only the \emph{one-shot} input/output data of the controlled plant. Thus, in contrast to conventional approaches, it does not require the iterative closed-loop control tests. The objective function of the proposed approach is derived by reformulating the ITAE minimization problem using a fictitious reference signal. The fictitious reference signal is constructed on the basis of the one-shot input/output data. The desired FO controller is obtained by minimizing the fictitious-reference-based objective function. A numerical example demonstrates that the proposed tuning technique can determine the FO controller for realizing the desired response in the same manner as the conventional approach while avoiding the iterative closed-loop tests.

Compared with the traditional approaches that focus on tuning of the FO controllers, the contribution and novelty of this study can be summarized as follows:
\begin{enumerate}[label=(C\arabic*)]
    \item \label{Contribution_MF}
    (\emph{Model-free}) The proposed approach tunes linear FO controllers using the input/output data of the controlled plant. Therefore, in contrast to traditional analytical approaches and optimization-based techniques via simulations, the proposed approach does not require mathematical models (e.g., the state-space model or transfer function model) of the controlled plant. Moreover, performance degradation due to the modeling error can be avoided, as the plant model itself is not required. 

    \item \label{Contribution_NI}
    (\emph{Non-iterative}) The required data collection is performed only once for the proposed approach. Thus, the proposed approach is free from the iterative closed-loop tests of the candidate controllers. Owing to its non-iterative nature, the proposed approach is superior to conventional optimization-based controller tuning, which involves multiple closed-loop tests. Thus, the proposed approach is simple and practical. 

    \item \label{Contribution_SI}
    (\emph{Simple implementation of FO control}) The model-free and non-iterative characteristics drastically reduce the development cost of FO controllers. Moreover, the proposed approach automatically searches for the desired FO controller by solving the numerical optimization problem; tedious manual controller tuning via trial-and-error is not required. Therefore, the proposed tuning technique can facilitate practical application of FO control by providing a novel controller tuning strategy. Note that controller tuning is a major challenge for implementing FO control in practice \cite{Tepljakov2021}.
\end{enumerate}

As stated in contributions \ref{Contribution_MF}--\ref{Contribution_SI}, this study provides a simple and practical implementation technique of linear FO controllers, bridging the gap between FO control theory and practice. Widespread use of FO control will substantially improve the performance of a wide variety of automatic control systems.

\subsection{Structure of the paper} \label{1_4_Structure}
The remainder of this paper is structured as follows. Section \ref{2_Preliminaries} summarizes some preliminaries. Section \ref{3_Proposed approach} introduces the proposed controller tuning approach. Section \ref{4_Numerical example} presents a numerical example to demonstrate the validity of the proposed approach. Finally, Section \ref{5_Conclusion} concludes the paper. 
%%%%%%%%%%%%%%%%%

%%% Section 2 %%%
\section{Preliminaries} \label{2_Preliminaries}

\subsection{Notation and symbols} \label{2_1_Notation}
The symbols $\mathbb{R}$ and $\mathbb{R}_{+}$ represent the sets of real numbers and strictly positive real numbers, respectively. 
The set of $n \times m$-dimensional matrices with real elements is denoted as $\mathbb{R}^{n \times m}$; let $\mathbb{R}^{n}$ denote $\mathbb{R}^{n \times 1}$ for simplicity.
For an $n$-dimensional real vector
$ x= \begin{bmatrix} x_{1} & x_{2} & \cdots & x_{n} \\ \end{bmatrix} ^ {\top}$,
let
$\mathrm{diag} \mleft( x \mright)$
denote the diagonal matrix, where the $i, i$-th component of 
$\mathrm{diag} \mleft( x \mright)$
is $x_{i}$.
Finally, the $p$-norm $\left\lVert x \right\rVert_{p}$ of $x$ is defined as 
$\left\lVert x \right\rVert_{p} \triangleq \mleft( \sum_{i=1}^{n} \left\lvert x_{i} \right\rvert ^{p} \mright)^\frac{1}{p}$.

The symbol $s$ represents the Laplace variable, whereas the $z$-variable is denoted as $z$. 
Let $\mathcal{R} \mleft[z\mright]$ denote the set of proper rational DT SISO LTI transfer functions. 
If a continuous-time (CT) transfer function $G\mleft( s \mright)$ has a tunable parameter $\phi$, we represent this parametrized transfer function explicitly as $G\mleft( s; \phi \mright)$. The analogous notation is used for DT transfer functions. 
For $G\mleft( s \mright)$, $\mathcal{Z} \mleft( G\mleft( s \mright) \mright)$ denotes its discretization. 
For a FO CT transfer function $G_{FO} \mleft( s \mright) $, 
$\mathcal{I} \mleft( G_{FO} \mleft( s \mright)  \mright)$
represents its IO-approximation (various IO-approximation techniques are available \cite{Deniz2020}).
Moreover, let
$\mathcal{Z I} \mleft( G_{FO} \mleft(s \mright)  \mright) \in \mathcal{R} \mleft[z\mright]$ 
denote the DT counterpart of $G_{FO} \mleft(s \mright)$, 
which can be obtained via, e.g., 
$\mathcal{Z}\mleft( \mathcal{I} \mleft( G_{FO}\mleft( s \mright) \mright)    \mright)$
or 
the impulse invariant discretization method \cite{Li2011}.

We use $\ast$ to denote the convolution operation. 
Specifically, the output $y = \mleft\{ y_{k} \mright\}_{k=0}^{N}$
of $G \mleft( z \mright) \in \mathcal{R} \mleft[ z \mright]$ due to an input $u = \mleft\{ u_{k} \mright\}_{k=0}^{N}$ is computed as
$y_{k} = G\mleft(z\mright) \ast u_{k} \triangleq \sum_{i=0}^{k} g_{i}u_{k-i} = \sum_{i=0}^{k} g_{k-i}u_{i} $,
where $\mleft\{ g_{k} \mright\}_{k=0}^{\infty} $ is the impulse response of $G\mleft(z\mright)$.

\subsection{FO controller tuning via ITAE minimization} \label{2_2_Tuning via ITAE minimization}
Figure \ref{Fig_ITAE_Problem} shows an overview of the FO controller tuning approach via ITAE minimization. In Figure \ref{Fig_ITAE_Problem}, $P \mleft(z\mright) \in \mathcal{R}\mleft[z\mright] $ is the plant to be controlled. The FO controller $C \mleft(z; \phi \mright) \in \mathcal{R}\mleft[z\mright] $ is given by 
$C \mleft( z; \phi \mright) = \mathcal{Z I} \mleft( C_{FO} \mleft( s; \phi \mright) \mright)$
, where the FO CT transfer function $C_{FO} \mleft(s; \phi \mright)$ has the tunable parameter $\phi$.
Let
$T \mleft(z; \phi \mright) \triangleq P\mleft( z \mright)C\mleft(z; \phi \mright) \mleft\{ 1 + P\mleft(z\mright) C\mleft(z; \phi\mright) \mright\}^{-1}$.
The setpoint reference, control input, and output are represented as 
$r = \mleft\{ r_{k} \mright\}_{k=0}^{\infty}$, 
$u = \mleft\{ u_{k} \mright\}_{k=0}^{\infty}$, 
$y = \mleft\{ y_{k} \mright\}_{k=0}^{\infty}$, 
respectively (these signals are evaluated for finite time steps in practice).

\begin{figure}[]%% placement specifier
    \centering%% For centre alignment of image.
    \includegraphics[width=0.45\textwidth]{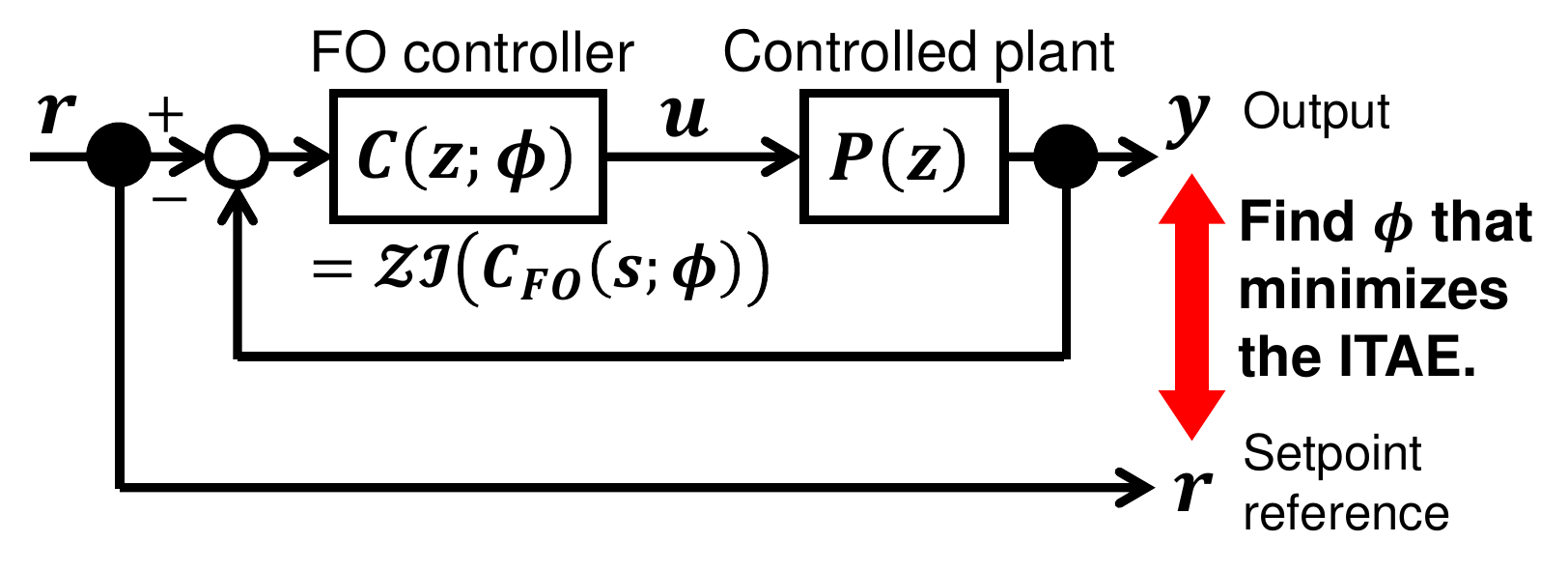}
    \caption{\protect\input{./Figures/Fig_1_Caption.tex}} 
    \label{Fig_ITAE_Problem}
\end{figure}  % \label{Fig_ITAE_Problem}

Let $r$ be the step signal, i.e., $r_{k} = r^{0}$ for all $k$, where $r^{0} \in \mathbb{R}$. The ITAE performance index $J_{ITAE}$ in the DT manner is defined as
    \begin{equation} \label{Eq_J_ITAE_Conventional}
        J_{ITAE} \mleft( \phi \mright) \triangleq \sum_{k=0}^{N} \mleft( k\tau_{s} \mright) \left\lvert y_{k} \mleft( \phi \mright) - r^{0} \right\rvert 
,
    \end{equation}
    \begin{equation} \label{Eq_yk_T_rk}
        y_{k} \mleft( \phi \mright) = T \mleft( z; \phi \mright) \ast r_{k}
,
    \end{equation}
where $\tau_{s} \in \mathbb{R}_{+}$ is the sampling time and $N$ indicates the maximum evaluation step.  Note that the output $y_{k} \mleft( \phi \mright)$ is governed by the controller parameter $\phi$. Then, the parameter $\phi$ of the FO controller is tuned such that $J_{ITAE} \mleft( \phi \mright)$ is minimized.

Minimization of the ITAE is the predominant tuning approach for controllers, including FO controllers (e.g., \cite{Lu2023, Paliwal2022}). The ITAE criterion applies less penalty to errors in the early stage and increases the penalty as time progresses. This design choice mitigates the impact of early-stage tracking errors on the performance index and emphasizes the influence of errors occurring at the final stage. Placing less emphasis on early-stage errors helps to suppress overly aggressive control actions in response to unavoidable initial tracking errors, thereby reducing overshoot and undershoot. The strict evaluation of final-state errors helps reduce steady-state errors. Therefore, minimizing the ITAE results in reducing overshoot and undershoot and steady-state errors, yielding practically desirable closed-loop behavior. The details of various integral performance criteria, including the ITAE, have been discussed elsewhere (e.g., \cite{Das2011, Tavazoei2010, Dorf2021}).

Figure \ref{Fig_Flowchart_Conventioal} shows the conventional procedure for minimizing the ITAE using an optimization algorithm. Clearly, the traditional approach shown in Figure \ref{Fig_Flowchart_Conventioal} requires multiple closed-loop tests to evaluate the ITAE. Iterative closed-loop tests using a real-world experimental device are tedious. Executing model-based simulations for the closed-loop tests requires costly and burdensome plant modeling, and the inevitable modeling error causes unexpected performance degradation. Therefore, a novel paradigm is required for controller tuning via ITAE minimization in order to avoid the iterative closed-loop control tests. 

\begin{figure}[]%% placement specifier
    \centering%% For centre alignment of image.
    \includegraphics[width=0.42\textwidth]{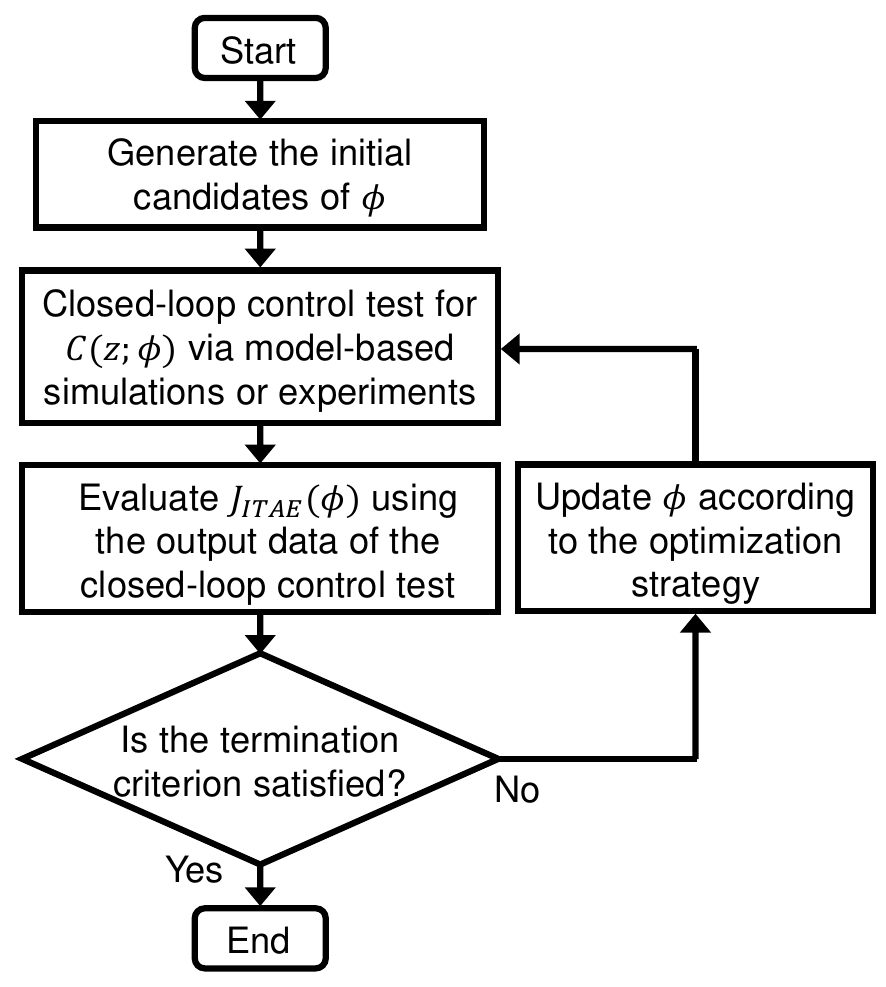}
    \caption{\protect\input{./Figures/Fig_2_Caption.tex}} 
    \label{Fig_Flowchart_Conventioal}
\end{figure}  % \label{Fig_Flowchart_Conventioal}
%%%%%%%%%%%%%%%%%

%%% Section 3 %%%
\section{Proposed approach} \label{3_Proposed approach}
To overcome the challenge posed by the conventional ITAE minimization approach, this section introduces the novel FO controller tuning strategy that can be executed using only one-shot input/output data of the controlled plant.

\subsection{Problem formulation} \label{3_1_Problem formulation}
This section states the FO controller tuning problem addressed in this study: ITAE minimization using only one-shot input/output data.  Here, it is assumed that we have the input data 
$u_{[0:N]}^D = \mleft\{ u_{k}^{D} \mright\}_{k=0}^{N}$
and the output data
$y_{[0:N]}^D = \mleft\{ y_{k}^{D} \mright\}_{k=0}^{N}$
that are obtained from either closed-loop or open-loop experiments. 
Without loss of generality, data collection is assumed to start when a nonzero input is applied to the controlled plant; that is, we assume $u_{0}^{D} \neq 0$.
The mathematical model of the controlled plant $P \mleft( z \mright)$ is assumed to be unknown. Then, we aim to find the optimal controller parameter $\phi^{\star}$ that minimizes the ITAE performance index \eqref{Eq_J_ITAE_Conventional} using 
$u_{[0:N]}^D$ and $y_{[0:N]}^D$.

\subsection{ITAE minimization using one-shot data: fictitious reference approach} \label{3_2_FR_ITAE_min}
This section describes the proposed FO controller tuning approach: fictitious-reference-based ITAE minimization (FR-ITAE-min). Specifically, controller tuning by minimizing the ITAE performance index \eqref{Eq_J_ITAE_Conventional} using 
$u_{[0:N]}^D$ and $y_{[0:N]}^D$ is as follows:
%% Theorem 1 %%
\begin{theorem}[FR-ITAE-min] \label{Theorem_FR_ITAE_min}
    The ITAE performance index \eqref{Eq_J_ITAE_Conventional} has the following data-driven representation:
     \begin{equation}\label{Eq_J_JD}
         J_{ITAE} \mleft( \phi \mright) = J_{ITAE}^{D} \mleft( \phi \mright)
,
     \end{equation}
     \begin{equation} \label{Eq_JD}
      \begin{split}  
        & J_{ITAE}^{D} \mleft( \phi \mright) \\
&= 
\left\lVert 
\mathrm{diag} \mleft( \bm{\tau}_{[0:N]} \mright)
\mleft\{
\bm{R^{0}} \mleft( \bm{\tilde{R}^{D}}\mleft( \phi \mright) \mright)^{-1}\bm{y^{D}}_{[0:N]}
-
r^{0} \bm{1}_{[0:N]}
\mright\}
\right\rVert_{1} 
,
      \end{split}
     \end{equation}
  where
    \begin{align}
        \bm{\tau}_{[0:N]} &=
\begin{bmatrix}
    0 & \tau_{s} & 2\tau_{s} & \cdots & N\tau_{s} 
 \end{bmatrix}^{\top}
, \label{Eq_tau_0_N} \\
        \bm{R^{0}} &=
\begin{bmatrix}
    r^{0}  & 0      & \cdots & 0 \\
    r^{0}  & r^{0}  & \ddots & 0 \\
    \vdots & \vdots & \ddots & 0 \\
    r^{0}  & r^{0}  & \cdots & r^{0}
 \end{bmatrix}
, \label{Eq_R_0} \\
        \bm{\tilde{R}^{D}}\mleft( \phi \mright) &=
\begin{bmatrix}
    \tilde{r}_{0}^{D} \mleft( \phi \mright)  & 0                    & \cdots & 0 \\
    \tilde{r}_{1}^{D} \mleft( \phi \mright) & \tilde{r}_{0}^{D} \mleft( \phi \mright)   & \ddots & 0 \\
    \vdots             & \vdots               & \ddots & 0 \\
    \tilde{r}_{N}^{D} \mleft( \phi \mright) & \tilde{r}_{N-1}^{D} \mleft( \phi \mright) & \cdots & \tilde{r}_{0}^{D} \mleft( \phi \mright)
 \end{bmatrix}
, \label{Eq_R_tilde} \\
        \tilde{r}_{k}^{D} \mleft( \phi \mright)
&= \mleft\{ C\mleft( z; \phi \mright) \mright\}^{-1} \ast u_{k}^{D} + y_{k}^{D}
, \label{Eq_r_k_tilde} \\
        \bm{y^{D}}_{[0:N]} &=
\begin{bmatrix}
    y_{0}^{D} & y_{1}^{D} & \cdots & y_{N}^{D} 
 \end{bmatrix}^{\top}
, \label{Eq_y_D_0_N} \\
        \bm{1}_{[0:N]} &=
\begin{bmatrix}
    1 & 1 & \cdots & 1 
 \end{bmatrix}^{\top}
, \label{Eq_ones_0_N} \\
    \end{align}
  provided that $\bm{\tilde{R}^{D}}\mleft( \phi \mright)$ is invertible. 
  Here, $\bm{R^{0}} \in \mathbb{R}^{\mleft(N+1 \mright) \times \mleft(N+1 \mright)}$ and $\bm{1}_{[0:N]} \in \mathbb{R}^{\mleft(N+1 \mright)}$.
  Then, we obtain the desired controller parameter $\phi^{\star}$ as
     \begin{equation} \label{Eq_argmin_JD}
       \phi^{\star} = \underset{\phi}{\operatorname{arg\,min}} \, J_{ITAE}^{D} \mleft( \phi \mright)
.
     \end{equation}
\end{theorem}
%%%%%%%%%%%%%%%

\begin{figure}[]%% placement specifier
    \centering%% For centre alignment of image.
    \includegraphics[width=0.42\textwidth]{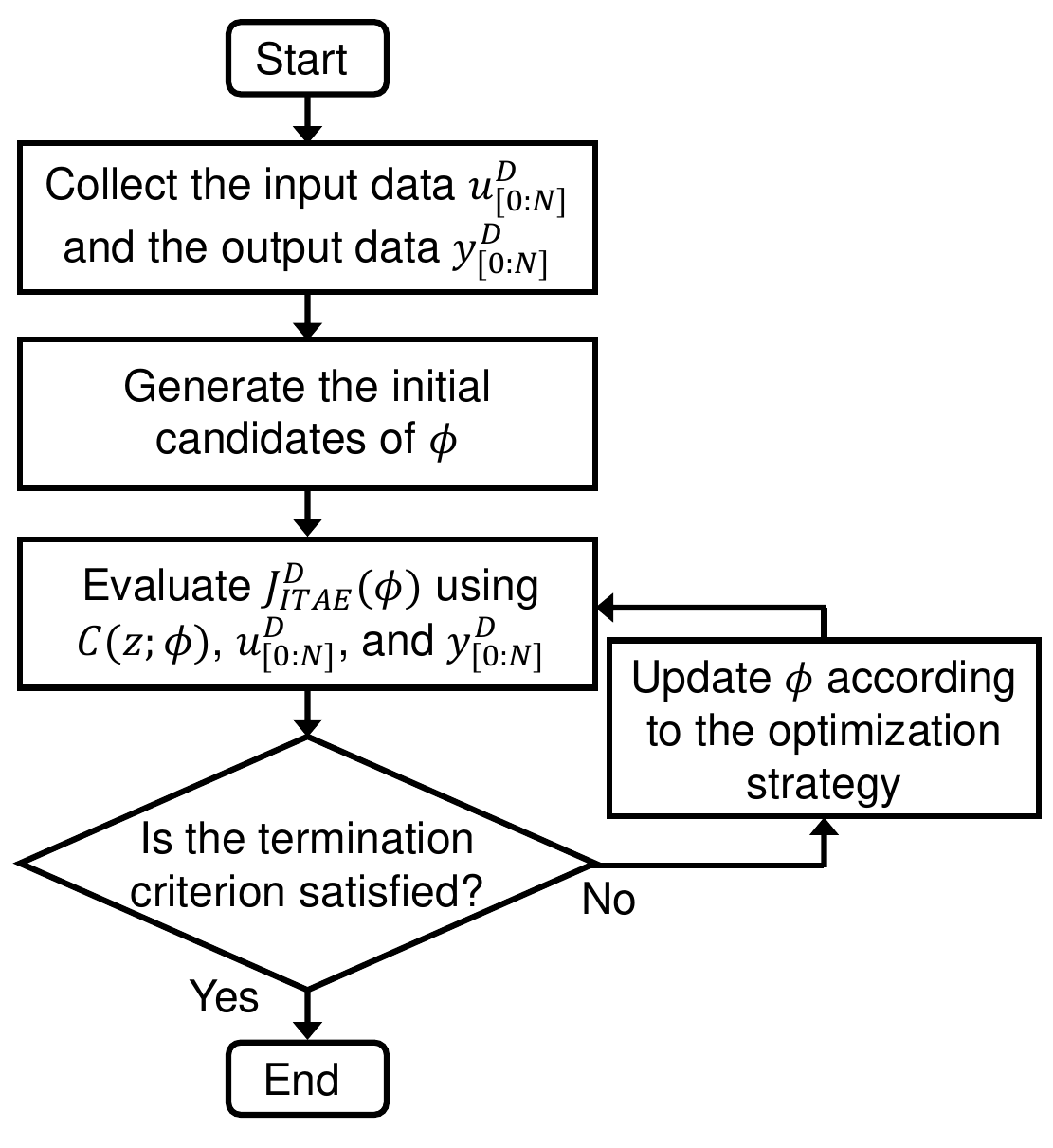}
    \caption{\protect\input{./Figures/Fig_3_Caption.tex}} 
    \label{Fig_Flowchart_Proposed}
\end{figure}  % \label{Fig_Flowchart_Proposed}

The signal
$\tilde{r}^{D} \mleft( \phi \mright) = \mleft\{ \tilde{r}_{k}^{D} \mleft( \phi \mright) \mright\}_{k=0}^{N}$
is the so-called \emph{the fictitious reference signal} \cite{Safonov1997}. The fictitious reference signal has the following property, which is used to prove Theorem \ref{Theorem_FR_ITAE_min}:
%% Lemma 1 %%
\begin{lemma} \label{Lemma_FicRef}
    $T \mleft( z; \phi \mright) \ast \tilde{r}_{k}^{D} \mleft( \phi \mright) = y_{k}^{D}$
    for any $\phi$, provided that
    $1 + P \mleft( z \mright) C \mleft(z; \phi\mright) \neq 0$.
\end{lemma}
%%%%%%%%%%%%%
\noindent
Lemma \ref{Lemma_FicRef} can be proved by some straightforward calculations. See, e.g., Eq. (8) in \cite{Kaneko2013}. Focusing on the property stated in Lemma \ref{Lemma_FicRef}, several studies have examined data-driven controller design approaches using the fictitious reference signal (e.g., \cite{Kaneko2013, Baldi2011, Yahagi2024, Yonezawa2024}), although none of them considers the ITAE performance.
In addition, due to Lemma \ref{Lemma_FicRef}, $\tilde{r}_{k}^{D} \mleft( \phi \mright)$ can be interpreted as an inversion-based reference input for $T \mleft( z; \phi \mright)$ when the target output is $y_{k}^{D}$. For example, provided that $P\mleft(z \mright)$ is biproper, we have 
$\tilde{r}_{k}^{D} \mleft( \phi \mright) 
= \mleft\{ P\mleft(z \mright) C\mleft( z; \phi \mright) \mright\}^{-1} \ast y_{k}^{D} + y_{k}^{D}
= \mleft\{ T \mleft( z; \phi \mright) \mright\}^{-1} \ast y_{k}^{D}$.
One study has proposed an inversion-based control technique for FO systems \cite{Padula2014}.

Now, we are ready to prove Theorem \ref{Theorem_FR_ITAE_min}.
%% Proof for Theorem 1 %%
\begin{proof}[Proof of Theorem \ref{Theorem_FR_ITAE_min}] \label{Proof_Thm_FR_ITAE_min}
    Let 
$t \mleft( \phi \mright) = \mleft\{ t_{k} \mleft( \phi \mright) \mright\}_{k=0}^{N}$
be the impulse response of $T \mleft( z; \phi \mright)$ and 
    \begin{equation} \label{Eq_t_phi_vec}
        \bm{t} \mleft( \phi \mright) 
\triangleq
\begin{bmatrix} 
    t_{0} \mleft( \phi \mright) & 
    t_{1} \mleft( \phi \mright) & \cdots & 
    t_{N} \mleft( \phi \mright) \\ 
\end{bmatrix} ^ {\top}
.
    \end{equation}
From \eqref{Eq_yk_T_rk}, we have 
    \begin{equation} \label{Eq_y_phi}
        \bm{y}\mleft( \phi \mright) \triangleq
\begin{bmatrix}
    y_{0}\mleft( \phi \mright) & y_{1}\mleft( \phi \mright) & \cdots & y_{N}\mleft( \phi \mright) 
 \end{bmatrix}^{\top}
=
\bm{R^0}\bm{t}\mleft( \phi \mright)
.
    \end{equation}
Then, \eqref{Eq_J_ITAE_Conventional} can be represented as
    \begin{align} \label{Eq_J_ITAE_1_norm}
        J_{ITAE} \mleft( \phi \mright) 
&=
\left\lVert 
    \mathrm{diag}\mleft( \bm{\tau}_{[0:N]} \mright) \mleft\{ \bm{y}\mleft( \phi \mright) - r^{0}\bm{1}_{[0:N]} \mright\}
\right\rVert_{1}  \notag \\
&=
\left\lVert 
    \mathrm{diag}\mleft( \bm{\tau}_{[0:N]} \mright) \mleft\{ \bm{R^{0}}\bm{t}\mleft( \phi \mright) - r^{0}\bm{1}_{[0:N]} \mright\}
\right\rVert_{1}
.
    \end{align}
Here, Lemma \ref{Lemma_FicRef} admits the data-driven representation of $\bm{t} \mleft( \phi \mright) $ as follows:
    \begin{equation} \label{Eq_t_phi_Ry}
        \bm{t}\mleft( \phi \mright) 
=
\mleft\{ \bm{\tilde{R}^{D}} \mleft( \phi \mright) \mright\}^{-1} \bm{y^{D}}_{[0:N]}
. 
    \end{equation}
Finally, substituting \eqref{Eq_t_phi_Ry} into \eqref{Eq_J_ITAE_1_norm} completes the proof.
\end{proof}
%%%%%%%%%%%%%%%%%%%%%%%%%

The implication of Theorem \ref{Theorem_FR_ITAE_min} is that we can evaluate the ITAE performance of the FO controller with \emph{arbitrary} $\phi$ using only the \emph{one-shot} input/output data. In other words, the ITAE minimization can be executed without iterative closed-loop control tests, as described in Figure \ref{Fig_Flowchart_Proposed}. A comparison of Figures \ref{Fig_Flowchart_Conventioal} and \ref{Fig_Flowchart_Proposed} clearly shows that the proposed approach is simpler than the conventional approach in that the repeated closed-loop tests are avoided in FR-ITAE-min. In Theorem \ref{Theorem_FR_ITAE_min}, the input/output data is assumed to be noiseless. If only noisy data is available, various denoising techniques can be used to reduce the effect of the noise before FR-ITAE-min is executed (e.g., total variation denoising \cite{Yahagi2024, Yahagi2021a}, discrete Fourier transform for periodically extended data \cite{Sakai2022}). The effect of the data noise on the proposed approach is discussed in \ref{Appendix_Measurement_Noise}.

 In Theorem \ref{Theorem_FR_ITAE_min}, it is assumed that $\bm{\tilde{R}^{D}}\mleft( \phi \mright)$ is invertible. This assumption is not restrictive in practical applications. Note that $\bm{\tilde{R}^{D}}\mleft( \phi \mright)$ has a lower triangular structure, as shown in \eqref{Eq_R_tilde}. The invertibility of $\bm{\tilde{R}^{D}}\mleft( \phi \mright)$ is determined by its diagonal element
$\tilde{r}_{0}^{D} \mleft( \phi \mright) =\mleft\{ C\mleft( z; \phi \mright) \mright\}^{-1} \ast u_{0}^{D} + y_{0}^{D} = c_{0}^{inv}\mleft( \phi \mright) u_{0}^{D} + y_{0}^{D}$
, where $\mleft\{ c_{k}^{inv}\mleft( \phi \mright) \mright\}_{k=0}^{\infty}$ denotes the impulse response of $\mleft\{ C\mleft( z; \phi \mright) \mright\}^{-1}$ (note that $c_{0}^{inv}\mleft( \phi \mright) \neq 0$ because $C\mleft( z; \phi \mright)$ is proper). Thus, $\bm{\tilde{R}^{D}}\mleft( \phi \mright)$ is not invertible if and only if $y_{0}^{D} = -c_{0}^{inv}\mleft( \phi \mright) u_{0}^{D}$, i.e., $\tilde{r}_{0}^{D} \mleft( \phi \mright) = 0$. Here, as $y_{[0:N]}^{D}$ is the output of the controlled plant $P \mleft( z \mright)$ due to $u_{[0:N]}^{D}$, $y_{0}^{D} =P \mleft(z \mright) \ast u_{0}^{D} = p_{0} u_{0}^{D}$ holds, where $\mleft\{ p_{k} \mright\}_{k=0}^{\infty}$ denotes the impulse response of the controlled plant $P \mleft( z \mright)$. 
These facts confirm that $\bm{\tilde{R}^{D}}\mleft( \phi \mright)$ is not invertible only if $p_{0} = -c_{0}^{inv}\mleft( \phi \mright)$. In other words, the proposed approach is infeasible only for the values of $\phi$ that satisfy $p_{0} = -c_{0}^{inv}\mleft( \phi \mright)$. Note that the invertibility of $\bm{\tilde{R}^{D}}\mleft( \phi \mright)$ can be tested by evaluating the value of $\tilde{r}_{0}^{D} \mleft( \phi \mright) $, which can be conducted without knowing $p_{0}$ itself.
Moreover, $\phi$ that makes $\bm{\tilde{R}^{D}}\mleft( \phi \mright)$ nearly singular (i.e., $\tilde{r}_{0}^{D} \mleft( \phi \mright) $ is very small) should be excluded from consideration in practice, as it leads to undesired closed-loop behavior. Specifically, such a $\phi$ may result in a large value of $t_{0} \mleft( \phi \mright)$ when $y_{0}^{D}\neq 0$, since $t_{0} \mleft( \phi \mright) = \mleft( \tilde{r}_{0}^{D} \mleft( \phi \mright) \mright)^{-1} y_{0}^{D}$. This relationship follows from \eqref{Eq_t_phi_Ry} and the fact that $\mleft(\bm{\tilde{R}^{D}}\mleft( \phi \mright) \mright)^{-1}$ is a lower triangular matrix whose diagonal elements are equal to $\mleft( \tilde{r}_{0}^{D} \mleft( \phi \mright) \mright)^{-1}$. An excessively large value of $t_{0} \mleft( \phi \mright)$ results in an overly aggressive response of the closed-loop system, which should be avoided in practice. That is, $\phi$ that makes $\bm{\tilde{R}^{D}}\mleft( \phi \mright)$ nearly singular should not be used from a practical perspective. In summary, it is rare to encounter the singularity of $\bm{\tilde{R}^{D}}\mleft( \phi \mright)$; we can exclude the values of $\phi$ that make $\bm{\tilde{R}^{D}}\mleft( \phi \mright)$ nearly singular from a practical perspective. Consequently, the invertibility assumption of $\phi$ that make $\bm{\tilde{R}^{D}}\mleft( \phi \mright)$ does not pose a significant restriction when applying the proposed approach.

In \eqref{Eq_JD}, the simple time weight matrix $\mathrm{diag} \mleft( \bm{\tau}_{[0:N]} \mright)$ may lead to numerical issues when $N$ (i.e., the number of data points) is large. One simple technique for addressing this issue is to modify the time weights. For example, to avoid the numerical sensitivity issue caused by the rapid growth of time weights for large $N$, $J_{ITAE}^{D} \mleft( \phi \mright)$ can be modified to 
\begin{equation} \label{Eq_JD_Sat}
    \begin{split} 
        & J_{ITAE, sat}^{D} \mleft( \phi \mright) \\
& \triangleq  
\left\lVert 
\mathrm{diag} \mleft( \mathrm{Sat} \mleft(\bm{\tau}_{[0:N]}\mright) \mright)
\mleft\{
\bm{R^{0}} \mleft( \bm{\tilde{R}^{D}}\mleft( \phi \mright) \mright)^{-1}\bm{y^{D}}_{[0:N]}
-
r^{0} \bm{1}_{[0:N]}
\mright\}
\right\rVert_{1} 
,
    \end{split}
\end{equation}
where
\begin{align}
    \mathrm{Sat} \mleft(\bm{\tau}_{[0:N]}\mright)
& \triangleq 
\begin{bmatrix}
    \mathrm{Sat} \mleft(0 \mright) & \mathrm{Sat} \mleft( \tau_{s} \mright)  & \cdots & \mathrm{Sat} \mleft( N\tau_{s} \mright)
\end{bmatrix}^{\top}
, \label{Eq_Sat_tau_0_N} \\
    \mathrm{Sat} \mleft( \tau \mright)
& \triangleq 
\mleft( \frac{2}{\pi}\alpha \mright) \arctan \mleft( \frac{\pi}{2\alpha}\tau \mright)
, \label{Eq_Sat_tau} \\
\end{align}
and $\alpha$ is a user-specified positive scalar. The time weight $\mathrm{Sat} \mleft( \tau \mright)$ can address the numerical problem because its range is restricted, i.e., $\mathrm{Sat} \mleft( \tau \mright) < \alpha$ for nonnegative $\tau$. On the other hand, $\mathrm{Sat} \mleft( \tau \mright)$ is a good approximation of the original time weight for small $\tau$, because $\mathrm{Sat} \mleft( \tau \mright)$ is a monotonically increasing function and
\begin{align}
    \mathrm{Sat} \mleft( \tau \mright)
&\approx 
\mathrm{Sat} \mleft( 0 \mright) 
+ \left. \mleft( \frac{d}{d\tau} \mathrm{Sat} \mleft( \tau \mright) \mright)  \right|_{\tau=0 } \tau  \\
&= \tau
. \label{Eq_Sat_tau_Approximation}
\end{align}
Therefore, $\mathrm{Sat} \mleft( k\tau_{s} \mright)$ approximates the original time weight $k\tau_{s}$ in the earlier stage (i.e., the time step $k$ is small), while the rapid growth of the time weight is avoided caused by the increase in $N$. 
In addition, the proposed fictitious-reference-based approach can be extended to minimize the IAE; see \ref{Appendix_IAE} for details.

\begin{remark} \label{Remark_Versatility}
    The conventional analytical approaches have been developed for specific (typically very simple) classes of controlled plants and FO controllers. Meanwhile, FR-ITAE-min is applicable to a wide variety of the controlled plants and controllers including high-order plants.
Such versatility stems from the fact that the proposed approach relies not on mathematical models of controlled plants but on input and output data. 
From this viewpoint, the proposed approach is a highly versatile and powerful tuning technique for FO controllers. 
In addition, the present approach can also be directly applied to tuning IO controllers, such as the IO-PID controller.
Note that the proposed approach relies on the fact that the output of a linear system is given by the convolution of the input and the impulse response of the system. Thus, the proposed approach is not applicable to systems with strong nonlinearity. The extension of the proposed approach to nonlinear systems remains future work.

\end{remark}

\begin{remark} \label{Remark_Reliability}
    The traditional analytical approaches tune the ideal FO controller $C_{FO} \mleft( s; \phi \mright)$, while ignoring the effect of $\mathcal{Z I} \mleft( \cdot  \mright)$. Note that the transformation to the implementation form affects the control performance and stability \cite{Deniz2020}. Here, the proposed approach evaluates the ITAE performance for the DT counterpart 
$C \mleft( z; \phi \mright) = \mathcal{Z I} \mleft( C_{FO} \mleft( s; \phi \mright) \mright)$ 
of the FO controller. In other words, FR-ITAE-min tunes the \emph{ready-to-implement} FO controller. 
This feature enables us to be free from explicit consideration of the approximation error involved in transforming ideal FO controllers into their implementable form.
Therefore, the proposed approach allows the reliable implementation of the FO controller.
Moreover, the proposed approach accepts various approximation methods, which makes the present approach highly versatile.
\end{remark}

\begin{remark} \label{Remark_Convexity}
    The main drawback of the proposed approach is that the optimization problem \eqref{Eq_argmin_JD} is non-convex as in the case of traditional ITAE minimization problem. In the future, a convex relaxation of \eqref{Eq_argmin_JD} will be investigated. 
Here, whether there exists any convex relaxation/approximation or not significantly depends on the parametrization of the FO controller to be tuned. To exploit the advantage of linear FO controllers, users must determine not only their gain but also their integral/derivative order. This requirement implies that the convex relaxation or approximation is not straightforward, because the exponent $\gamma$ of the FO integrator/differentiator $s^{\gamma}$ in the Laplace domain serves as a tuning parameter.
Meanwhile, in the numerical example in this study, \eqref{Eq_argmin_JD} can be successfully solved via typical metaheuristic algorithms implemented in a commercial numerical software. 
It should be pointed out that, although the optimization problem \eqref{Eq_argmin_JD} for the proposed approach is non-convex, the cost for solving it is significantly reduced compared with traditional optimization-based FO controller tuning approaches. It is because, unlike the conventional approaches, the proposed approach does not require actually executing closed-loop control tests to evaluate the objective function value. 
\end{remark}

%%%%%%%%%%%%%%%%%

%%% Section 4 %%%
\section{Numerical example} \label{4_Numerical example}

\subsection{Overview of the validation} \label{4_1_Overview_validation}
\begin{table*}[]    
    \centering
    \begin{threeparttable}[h]
      \caption{\protect Tuning strategies for the numerical example.} 
      \label{Table_Strategies} 
      \begin{tabular}{cccc}
        \hline
        \begin{tabular}{c}
\end{tabular} & \begin{tabular}{c}
    Exp-ITAE-min
\end{tabular} & \begin{tabular}{c}
    MB-ITAE-min
\end{tabular} & \begin{tabular}{c}
    FR-ITAE-min
\end{tabular} \\\hline
        \begin{tabular}{c}
Controlled plant for \\ closed-loop test
\end{tabular} & $P_{full} \mleft( z \mright)$ & $P_{reduced} \mleft( z \mright)$ & $P_{full} \mleft( z \mright)$ \\\hline
        \begin{tabular}{c}
Objective function
\end{tabular} & $J_{ITAE} \mleft( \phi \mright)$ & $J_{ITAE} \mleft( \phi \mright)$ & $J_{ITAE}^{D} \mleft( \phi \mright)$ \\\hline
        \begin{tabular}{c}
Purpose of the \\ closed-loop test
\end{tabular} & \begin{tabular}{c}
Objective function \\ evaluation
\end{tabular} & \begin{tabular}{c}
Objective function \\ evaluation
\end{tabular} & \begin{tabular}{c}
Initial data collection
\end{tabular} \\\hline
        \begin{tabular}{c}
Required number \\ of closed-loop tests
\end{tabular} & \begin{tabular}{c}
$4.5 \times 10^{4}$
\end{tabular}\tnote{$\dag$} & \begin{tabular}{c}
$4.5 \times 10^{4}$
\end{tabular}\tnote{$\dag$} & \begin{tabular}{c}
$1$
\end{tabular}\tnote{$\ddag$} \\\hline
        \begin{tabular}{c}
Implication
\end{tabular}
 & \begin{tabular}{c}
    Conventional \\ 
    approach based on \\ 
    repeated real-world \\ 
    closed-loop experiments
\end{tabular} & \begin{tabular}{c}
    Conventional \\ 
    approach based on\\
    iterative model-based \\ 
    simulations 
\end{tabular} & \begin{tabular}{c}
    Proposed approach \\ 
    with neither repeated \\ 
    experiments nor \\ 
    model-based simulations
\end{tabular} \\
        \hline
      \end{tabular}
    \begin{tablenotes}
    \item[$\dag$] Same as the total number of objective function evaluations
    \item[$\ddag$] Same as total number of initial experiments for data collection
    \end{tablenotes}
    \end{threeparttable}
\end{table*}
%%%%%%%%%%%%%%%%%%%%%% %\label{Table_Strategies} 

The validity of the proposed approach is demonstrated by a numerical study. Equation \eqref{Eq_P_full_z} describes the controlled plant $P_{full} \mleft( z \mright)$ in this example, which is a highly oscillatory high-order process system \cite{Wang2000, Das2018}:
    \begin{equation} \label{Eq_P_full_z}
        P_{full} \mleft(z \mright) = \mathcal{Z} \mleft( P_{full} \mleft( s \mright) \mright)
,
    \end{equation}
    \begin{equation}\label{Eq_P_full_s}
    P_{full} \mleft( s \mright) = \frac{1}{ \mleft(s^{2} + s + 1 \mright) \mleft(s + 3 \mright) } e^{-s}
,
    \end{equation}
where the Tustin method is adopted for $\mathcal{Z} \mleft( \cdot \mright)$ with the sampling time $\tau_{s} = 1.0\times10^{-2} \, \si{s}$. Note that $P_{full} \mleft( z \mright)$ serves as the actual controlled plant, and the exact model is assumed to be unavailable. As the controller $C \mleft(z; \phi \mright)$, we use the FO-PID controller as follows \cite{Dastjerdi2019}:
    \begin{equation} \label{Eq_C_z_phi}
        C \mleft( z; \phi \mright) = \mathcal{Z} \mleft( \mathcal{I} \mleft( C_{FOPID} \mleft( s; \phi \mright)  \mright)  \mright)
,
    \end{equation}
    \begin{equation} \label{Eq_C_FOPID_s}
        C_{FOPID} \mleft( s; \phi \mright) =
K_{fp}
+
K_{fi} \frac{1}{s^{\lambda}}
+
K_{fd} \frac{s^{\mu}}{1+\tau_{f}s^{\mu}}
,
    \end{equation}
    \begin{equation} \label{Eq_phi_design}
        \phi = 
\begin{bmatrix}
    K_{fp} & K_{fi} & \lambda & K_{fd} & \mu
\end{bmatrix}^{\top}
,\>
\tau_{f} = \tau_{s}
,

    \end{equation}
where the search ranges of $\phi$ are 
$K_{fp} \, K_{fi} \, K_{fd} \in \mleft[0 ,\> 10 \mright]$
and
$\lambda \, \mu \in \mleft[0 ,\> 2 \mright]$. For $\mathcal{Z} \mleft( \cdot \mright)$ in \eqref{Eq_C_z_phi}, the Tustin method is adopted, where $\tau_{s} = 1.0\times10^{-2} \, \si{s}$. The Oustaloup recursive filter \cite{Oustaloup2000} is applied to $\mathcal{I} \mleft( \cdot \mright)$ in \eqref{Eq_C_z_phi} 
(order: $5$, valid frequency range: $\mleft(1.0\times10^{-6} ,\, 10^{3} \mright) \, \si{rad/s}$). The IO approximation via the Oustaloup filter is performed using the FOMCON toolbox \cite{Tepljakov2011}.

In this example, the proposed approach, FR-ITAE-min, is compared with the traditional ITAE minimization strategy shown in Figure \ref{Fig_Flowchart_Conventioal}. Table \ref{Table_Strategies} summarizes the different controller tuning techniques. We employ the PSO algorithm implemented in the MATLAB Global Optimization Toolbox as the optimization solver, where the total number of the objective function evaluations is $4.5 \times 10^{4}$ and the population size is set to $150$.

In Table \ref{Table_Strategies}, Exp-ITAE-min is the conventional ITAE minimization approach with $P_{full} \mleft( z \mright)$ for the closed-loop tests with $r^{0} = 1$; it mimics the ITAE minimization procedure based on the closed-loop experiments using the real-world controlled plant. MB-ITAE-min is the conventional approach with $r^{0} = 1$ and $P_{reduced} \mleft( z \mright)$ shown in \eqref{Eq_P_reduced_z} for the closed-loop tests. Here, $P_{reduced} \mleft( s \mright)$ shown in \eqref{Eq_P_reduced_s} is the reduced-order model of $P_{full} \mleft( s \mright)$ in \eqref{Eq_P_full_s} \cite{Wang2000, Das2018}. MB-ITAE-min mimics the ITAE minimization strategy based on the closed-loop tests via the model-based simulations with $P_{reduced} \mleft( z \mright)$. In Exp-ITAE-min and MB-ITAE-min, the closed-loop tests are executed for $25 \, \si{s}$.
    \begin{equation} \label{Eq_P_reduced_z}
        P_{reduced} \mleft(z \mright) = \mathcal{Z} \mleft( P_{reduced} \mleft( s \mright) \mright)
    \end{equation}
    \begin{equation} \label{Eq_P_reduced_s}
        P_{reduced} \mleft( s \mright) = \frac{1}{ 3.2158s^{2} + 3.1614s + 3.0568 } e^{-1.28s}
    \end{equation}

FR-ITAE-min, i.e., the proposed approach, uses the input/output data shown in Figure \ref{Fig_Example_IO_Data} as $u_{[0:N]}^D$ and $y_{[0:N]}^D$. The data shown in Figure \ref{Fig_Example_IO_Data} is collected by a closed-loop test for $P_{full} \mleft( z \mright)$, where the controller parameter is 
$\phi^{0} = \begin{bmatrix} 1 & 0 & 1 & 0 & 1 \end{bmatrix}^{\top}$
and the setpoint $r^{0} = 1$. Note that not only the data shown in Figure \ref{Fig_Example_IO_Data} but also various input/output data can be used to execute the proposed approach.

\begin{remark} \label{Remark_Innovation}
    The innovation of the proposed approach is not the application of the PSO algorithm for tuning the FO controller but the novel data-driven formulation of the ITAE-based controller tuning. In contrast to Exp-ITAE-min/MB-ITAE-min, which is the typical PSO-based tuning strategy, the proposed approach requires neither repeated experiments nor iterative model-based simulations. The proposed approach relies on only the one-shot input/output data; the PSO algorithm is merely the solver for the optimization problem based on the one-shot data. Hence, the purpose of this numerical example is to demonstrate the validity of the novel formulation of the ITAE minimization based on the one-shot data.
\end{remark}

\begin{remark} \label{Remark_Parameter}
    In the proposed approach, we can freely choose the controller parameter to be tuned. This is because the objective function \eqref{Eq_JD} of the proposed approach does not depend on any particular choice of the controller parameters. Thus, if some controller parameters are known in advance based on prior knowledge of the controlled plant, they can be fixed at those values, and the proposed approach can be used to tune the remaining parameters. 
For example, if the evaluation interval is sufficiently long, minimizing the ITAE tends to result in $\lambda \geq 1$, in order to reduce the settling time. This is because $\lambda \in \mleft(0, \, 1\mright)$ results in a long settling time \cite{Dastjerdi2019}, which causes the ITAE value to become very large over a long evaluation interval. However, the evaluation interval may not be long enough in practice. This situation can lead to $\lambda \in \mleft(0, \, 1\mright)$, which results in a long settling time \cite{Dastjerdi2019}, depending on the characteristics of the controlled plant. To circumvent this issue, $\lambda$ can be fixed at $1$ rather than treated as a tuning parameter in the proposed approach.
\end{remark}

\begin{figure}[]%% placement specifier
    \centering%% For centre alignment of image.
    \includegraphics[width=0.4\textwidth]{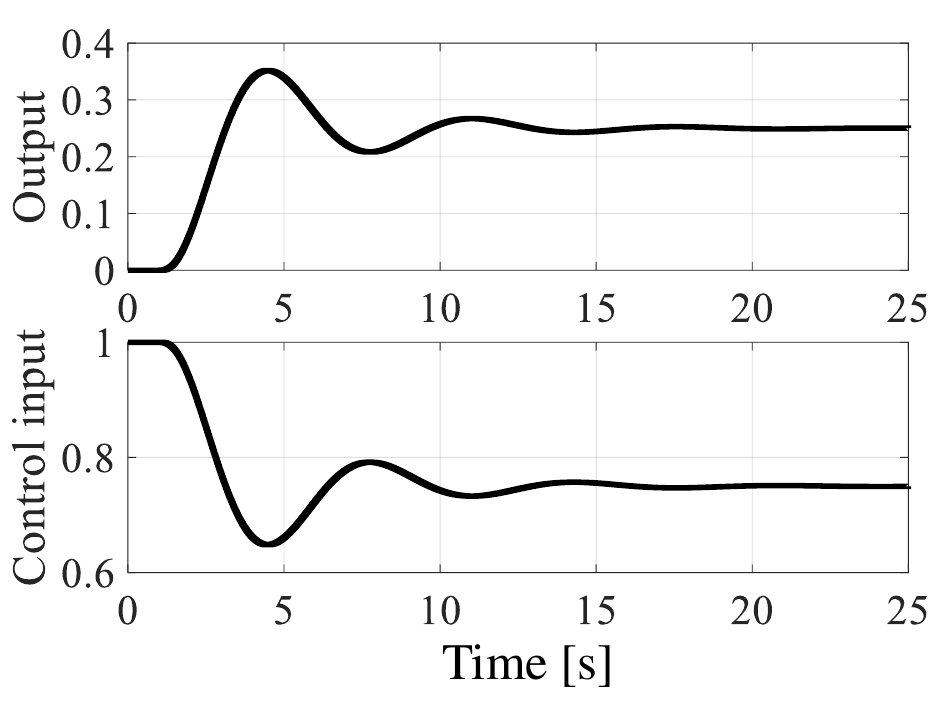}
    \caption{\protect\input{./Figures/Fig_4_Caption.tex}} \label{Fig_Example_IO_Data}
\end{figure}  % \label{Fig_Example_IO_Data}

\subsection{Results and discussion} \label{4_2_Results_discussion}
Table \ref{Table_Results} summarizes the tuning result of Exp-ITAE-min, MB-ITAE-min, and FR-ITAE-min. Here, $J_{ITAE}$ of Exp-ITAE-min is computed using $P_{full} \mleft( z \mright)$, whereas $J_{ITAE}$ of MB-ITAE-min is computed using $P_{reduced} \mleft( z \mright)$. In FR-IDFRIT-min, $J_{ITAE}^{D}$ is evaluated using the data $u_{[0:N]}^{D}$ and $y_{[0:N]}^{D}$ collected from $P_{full} \mleft( z \mright)$. Figure \ref{Fig_Example_Control_Result} shows the control results obtained using the controllers tuned by Exp-ITAE-min, MB-ITAE-min, and FR-ITAE-min. The black line represents the setpoint reference signal. The blue, green, and magenta lines represent the controlled results obtained using the controllers tuned via Exp-ITAE-min, MB-ITAE-min, and FR-ITAE-min, respectively.

\begin{table*}[]    
    \centering
    \begin{threeparttable}[h]
      \caption{\protect Summary of the tuning results. FR-ITAE-min (proposed) provides nearly the same controller as Exp-ITAE-min (conventional, experiment-based), whereas the required number of the closed-loop experiments for FR-ITAE-min is significantly smaller than that for Exp-ITAE-min. As for MB-ITAE-min (conventional, model-based), the ITAE performance for the actual controlled plant $P_{full}$ is inferior to those provided by Exp-ITAE-min and FR-ITAE-min because of the gap between $P_{full}$ and its model $P_{reduced}$ for controller tuning.} 
      \label{Table_Results} 
      \begin{tabular}{cccc}
        \hline
        \begin{tabular}{c}
\end{tabular} & \begin{tabular}{c}
    Exp-ITAE-min
\end{tabular} & \begin{tabular}{c}
    MB-ITAE-min
\end{tabular} & \begin{tabular}{c}
    FR-ITAE-min
\end{tabular} \\\hline
        \begin{tabular}{c}
$K_{fp}$
\end{tabular} & $2.4278$ & $2.3087$ & $2.4280$ \\
        \begin{tabular}{c}
$K_{fi}$
\end{tabular} & $1.4777$ & $1.4768$ & $1.4777$ \\
        \begin{tabular}{c}
$\lambda$
\end{tabular} & $1.0044$ & $1.0038$ & $1.0044$ \\
        \begin{tabular}{c}
$K_{fd}$
\end{tabular} & $2.2934$ & $2.2154$ & $2.2934$ \\
        \begin{tabular}{c}
$\mu$
\end{tabular} & $1.2270$ & $1.1974$ & $1.2271$ \\
        \begin{tabular}{c}
Objective function \\ 
value via $\phi^{0}$
\end{tabular} & \begin{tabular}{c}
    $J_{ITAE} \mleft( \phi^{0} \mright) = 2.3463 \times 10^{4}$
\end{tabular}\tnote{$\dag$} & \begin{tabular}{c}
    $J_{ITAE} \mleft( \phi^{0} \mright) = 2.3572 \times 10^{4}$
\end{tabular}\tnote{$\ddag$} & \begin{tabular}{c}
    $J_{ITAE}^{D} \mleft( \phi^{0} \mright) = 2.3463 \times 10^{4}$
\end{tabular} \\
        \begin{tabular}{c}
Objective function \\ 
value via $\phi^{\star}$
\end{tabular} & \begin{tabular}{c}
    $J_{ITAE} \mleft( \phi^{\star} \mright) = 2.7252 \times 10^{2}$
\end{tabular}\tnote{$\dag$} & \begin{tabular}{c}
    $J_{ITAE} \mleft( \phi^{\star} \mright) = 2.7578 \times 10^{2}$
\end{tabular}\tnote{$\ddag$} & \begin{tabular}{c}
    $J_{ITAE}^{D} \mleft( \phi^{\star} \mright) = 2.7252 \times 10^{2}$
\end{tabular} \\
        \begin{tabular}{c}
ITAE performance \\ 
for $P_{full} \mleft( z \mright)$
\end{tabular} & \begin{tabular}{c}
    $2.7252 \times 10^{2}$
\end{tabular} & \begin{tabular}{c}
    $2.7651 \times 10^{2}$
\end{tabular} & \begin{tabular}{c}
    $ 2.7252 \times 10^{2}$
\end{tabular} \\
        \hline
      \end{tabular}
      \begin{tablenotes}
        \item[$\dag$] Computed using the actual controlled plant $P_{full} \mleft( z \mright)$
        \item[$\ddag$] Computed using the reduced order model $P_{reduced} \mleft( z \mright)$ of $P_{full} \mleft( z \mright)$
      \end{tablenotes}
      \end{threeparttable}
\end{table*} % \label{Table_Results} 

\begin{figure}[]%% placement specifier
    \centering%% For centre alignment of image.
    \includegraphics[width=0.47\textwidth]{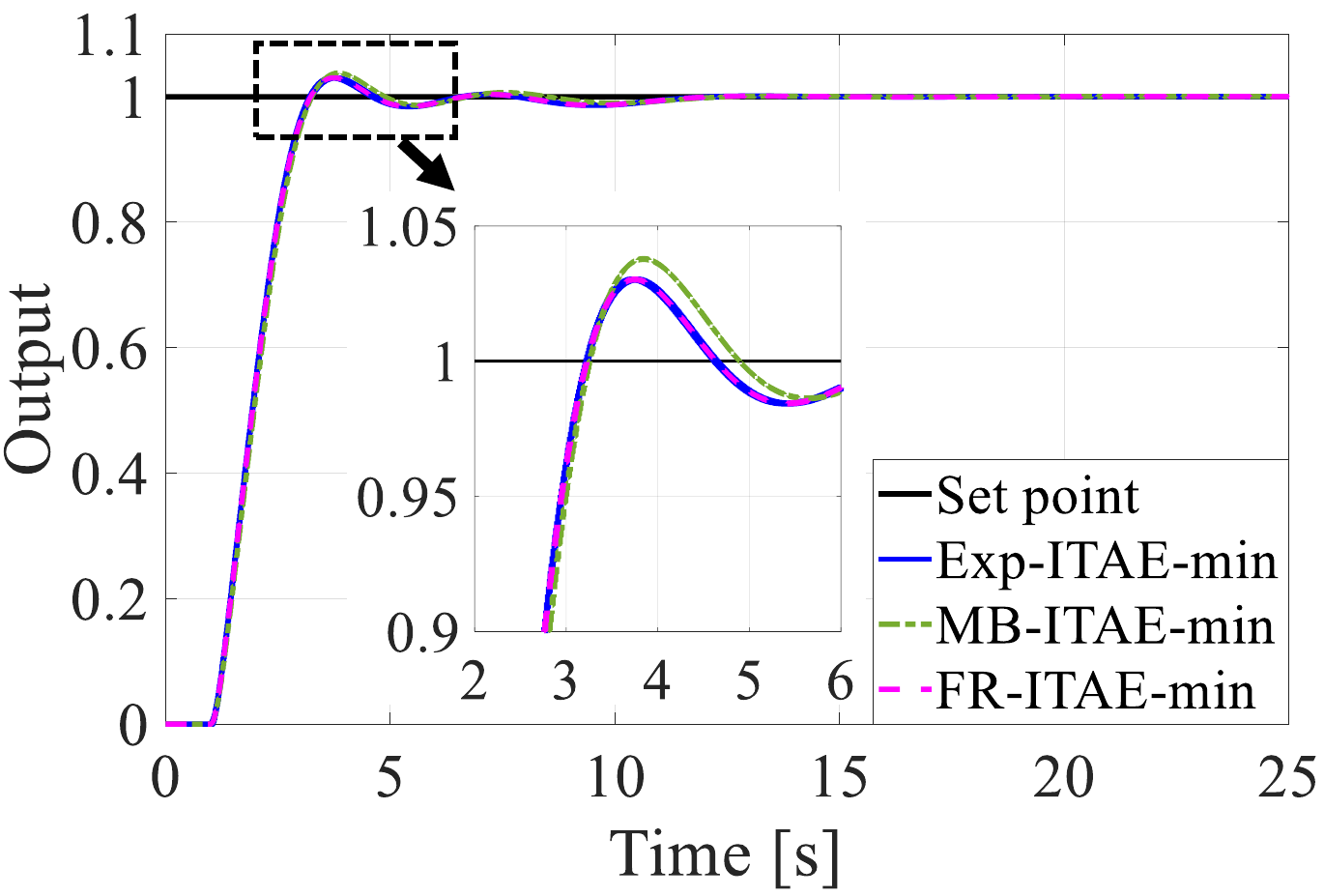}
    \caption{\protect\input{./Figures/Fig_5_Caption.tex}} 
    \label{Fig_Example_Control_Result}
\end{figure}  % \label{Fig_Example_Control_Result}

Table \ref{Table_Results} and Figure \ref{Fig_Example_Control_Result} clearly show that all the tuning strategies reduce the objective function values and provide controllers that achieve good control performance for $P_{full} \mleft(z \mright)$, i.e., the actual controlled plant. A comparison of the results of Exp-ITAE-min and FR-ITAE-min confirms that the proposed approach provides a controller that is nearly identical to that provided by the conventional strategy, which designs the ideal controller on the basis of the iterative closed-loop tests using the real-world control plant. In fact, Table \ref{Table_Results} demonstrates that the controller parameter provided by FR-ITAE-min is nearly the same as that provided by Exp-ITAE-min. However, FR-ITAE-min significantly reduces the required number of closed-loop experiments, compared with Exp-ITAE-min. Therefore, from the viewpoint of simplicity, the proposed approach is superior to the traditional approach involving repeated closed-loop experiments for the real-world control system. Here, note that the value of $J_{ITAE} \mleft( \phi^{0} \mright)$ in Exp-ITAE-min is exactly the same as that of $J_{ITAE}^{D} \mleft( \phi^{0} \mright)$ in FR-ITAE-min. This implies that the $J_{ITAE}^{D}$ is a valid one-shot data-driven representation of $J_{ITAE}$. The justification of $J_{ITAE}^{D}$ can be also provided from the fact that the values of $J_{ITAE}^{D} \mleft( \phi^{\star} \mright)$ and the ITAE performance for $P_{full} \mleft( z \mright)$ in FR-ITAE-min is exactly the same. 

In contrast to Exp-ITAE-min, the value of $J_{ITAE} \mleft( \phi^{\star} \mright)$ in MB-ITAE-min is different from the ITAE performance for $P_{full} \mleft( z \mright)$. This is because MB-ITAE-min relies on the reduced-order model $P_{reduced} \mleft( z \mright)$ of the actual controlled plant $P_{full} \mleft( z \mright)$. In other words, this gap is caused by the modeling error in the plant model that is employed for controller tuning. Note that the modeling error causes unexpected performance degradation. In fact, Figure \ref{Fig_Example_Control_Result} and the ITAE performance in Table \ref{Table_Results} show that the controller due to MB-ITAE-min suffers from a larger overshoot and worse ITAE performance than those due to Exp-ITAE-min and FR-ITAE-min. Meanwhile, FR-ITAE-min is driven by the one-shot data obtained from the actual controlled plant; it is free from the modeling error as the plant model itself is not required (i.e., the model-free strategy). Therefore, the proposed approach can achieve better control performance than the traditional model-based approaches for the actual controlled plant.

The proposed approach requires only one-shot data, thereby avoiding plant modeling and multiple experiments. In the proposed approach, the desired controller is automatically found by solving the optimization problem; it does not require manual trial-and-error by the designer. Moreover, the control performance of the FO controller is evaluated in the ready-to-implement form in the proposed approach, ensuring reliability of the resulting controller. Consequently, owing to these useful features, the proposed approach can drastically reduce the implementation burden of linear FO control, as controller tuning is a major challenge in the implementation of FO controllers. Widespread use of FO control will significantly improve the performance and reliability of practical automatic control systems. 

%%%%%%%%%%%%%%%%%

%%% Section 5 %%%
\section{Conclusion} \label{5_Conclusion}
This study proposed a novel tuning technique for linear FO controllers. The proposed approach provides the optimal controller in the ITAE sense. In contrast to the conventional approaches, the proposed approach performs ITAE minimization without repeated closed-loop tests, which is the most important novelty of this study.  To avoid multiple closed-loop tests, we reformulated the ITAE minimization problem using the fictitious reference signal that is computed on the basis of the one-shot input/output data and the controller to be tested. The validity of the proposed approach was demonstrated by a numerical study. The proposed approach can provide the same FO controller as the conventional ITAE minimization strategy based on multiple closed-loop experiments using the actual control system. Therefore, the proposed approach is superior to the conventional one, as it avoids multiple closed-loop tests. Moreover, the controller obtained using the present approach achieves better control performance than that due to the traditional simulation-based optimization relying on the plant model. This result confirms that the proposed approach is superior to the traditional model-based approach from the viewpoint of the unnecessity of the plant model. In summary, this study has provided a simple and practical controller tuning technique for linear FO controllers.

In the future, the selection of the optimization solver for the proposed approach will be explored.
%%%%%%%%%%%%%%%%%

%%% Some information of this work %%%
\section*{Conflict of interest}
The authors declare no potential conflict of interests.

\section*{Author contributions}
\textbf{Ansei Yonezawa}: Conceptualization; data curation; formal analysis; funding acquisition; investigation; methodology; software; validation; visualization; writing--original draft; writing--review \& editing.
\textbf{Heisei Yonezawa}: Conceptualization; methodology; investigation; writing--original draft; writing--review \& editing.
\textbf{Shuichi Yahagi}: Conceptualization; formal analysis; investigation; writing--review \& editing.
\textbf{Itsuro Kajiwara}: Conceptualization; funding acquisition; supervision; writing--review \& editing.
\textbf{Shinya Kijimoto}: Conceptualization; resource; writing--review \& editing.
%%%%%%%%%%%%%%%%%%%%%%%%%%%%%%%%%%%%%

\appendix    % Appendix
\section{Effect of measurement noise}  \label{Appendix_Measurement_Noise}
In Theorem \ref{Theorem_FR_ITAE_min}, it is assumed that the data used in the proposed approach is noiseless. However, measured data is often corrupted by noise in real-world systems. In this appendix, we analyze the effect of measurement noise on the proposed approach. Specifically, we consider noise affecting the output data, as the output of the controlled plant is typically obtained through sensors, which often introduce sensor noise. Let
$y_{[0:N]}^{nD} = \mleft\{ y_{k}^{nD} \mright\}_{k=0}^{N}$
and
$n_{[0:N]}^{D} = \mleft\{ n_{k}^{D} \mright\}_{k=0}^{N}$
denote the noisy output data and the measurement noise, respectively, where $y_{k}^{nD} \triangleq  y_{k}^{D} +  n_{k}^{D}$. Note that $y_{k}^{D}$ is the clean output of the controlled plant, i.e., $y_{k}^{D} = P \mleft( z \mright) \ast u_{k}^{D}$. Substituting $y_{[0:N]}^{nD} $ into \eqref{Eq_JD} in Theorem \ref{Theorem_FR_ITAE_min} yields the performance index $J_{ITAE}^{nD} \mleft( \phi \mright)$ of FR-ITAE-min based on the noisy output data:
\begin{multline} 
    J_{ITAE}^{nD} \mleft( \phi \mright) \\
= 
\left\lVert
\mathrm{diag} \mleft( \bm{\tau}_{[0:N]} \mright)
\mleft\{
\bm{R^{0}} \mleft( \bm{\tilde{R}^{nD}}\mleft( \phi \mright) \mright)^{-1}\bm{y^{nD}}_{[0:N]}
-
r^{0} \bm{1}_{[0:N]}
\mright\}
\right\rVert_{1} 
, \label{Eq_JD_Noisy}
\end{multline} 
where
\begin{align}
    \bm{\tilde{R}^{nD}}\mleft( \phi \mright) &=
\begin{bmatrix}
    \tilde{r}_{0}^{nD} \mleft( \phi \mright) & 0                    & \cdots & 0 \\
    \tilde{r}_{1}^{nD} \mleft( \phi \mright) & \tilde{r}_{0}^{nD} \mleft( \phi \mright)   & \ddots & 0 \\
    \vdots             & \vdots               & \ddots & 0 \\
    \tilde{r}_{N}^{nD} \mleft( \phi \mright) & \tilde{r}_{N-1}^{nD} \mleft( \phi \mright) & \cdots & \tilde{r}_{0}^{nD} \mleft( \phi \mright)
 \end{bmatrix}
, \label{Eq_R_tilde_Noisy} \\
    \tilde{r}_{k}^{nD} \mleft( \phi \mright)
&= \mleft\{ C\mleft( z; \phi \mright) \mright\}^{-1} \ast u_{k}^{D} + y_{k}^{nD}
, \label{Eq_r_k_tilde_Noisy} \\
    \bm{y^{nD}}_{[0:N]} &=
\begin{bmatrix}
    y_{0}^{nD} & y_{1}^{nD} & \cdots & y_{N}^{nD} 
 \end{bmatrix}^{\top}
. \label{Eq_y_D_0_N_Noisy} \\
\end{align}
Then, the following proposition characterizes the effect of the measurement noise on the proposed approach:
%%% Proposition 1
\begin{proposition} \label{Proposition_FR_ITAE_min_Noisy}
    The noisy performance index $J_{ITAE}^{nD} \mleft( \phi \mright)$ shown in \eqref{Eq_JD_Noisy} can be represented as
\begin{equation}
    J_{ITAE}^{nD} \mleft( \phi \mright)
= 
\left\lVert
\bm{\ell}_{ITAE}^{D}\mleft( \phi \mright)
+
\bm{d}\mleft( \phi, n_{[0:N]}^{D} \mright)
\right\rVert_{1} 
, \label{Eq_JD_ITAE_Noisy_L_and_D}
\end{equation}
where
\begin{equation}
    \bm{\ell}_{ITAE}^{D} \mleft( \phi \mright)
= 
\mathrm{diag} \mleft( \bm{\tau}_{[0:N]} \mright)
\mleft\{
\bm{R^{0}} \mleft( \bm{\tilde{R}^{D}}\mleft( \phi \mright) \mright)^{-1}\bm{y^{D}}_{[0:N]}
-
r^{0} \bm{1}_{[0:N]}
\mright\}
, \label{Eq_LD_ITAE}
\end{equation}
\begin{multline} 
    \bm{d}\mleft( \phi, n_{[0:N]}^{D} \mright) 
=
\mathrm{diag} \mleft( \bm{\tau}_{[0:N]} \mright) \bm{R^{0}} \mleft( \bm{\tilde{R}^{D}}\mleft( \phi \mright) \mright)^{-1}\bm{n^{D}}_{[0:N]} \\
- \bm{H}\mleft( \phi, n_{[0:N]}^{D} \mright) \mleft( \bm{\tilde{R}^{D}}\mleft( \phi \mright) \mright)^{-1} 
   \mleft(\bm{y^{D}}_{[0:N]} + \bm{n^{D}}_{[0:N]} \mright)
, \label{Eq_D_phi_n}
\end{multline} 
\begin{equation}
    \bm{n^{D}}_{[0:N]} =
\begin{bmatrix}
    n_{0}^{D} & n_{1}^{D} & \cdots & n_{N}^{D} 
 \end{bmatrix}^{\top}
, \label{Eq_n_D_0_N}
\end{equation}
\begin{multline}
    \bm{H}\mleft( \phi, n_{[0:N]}^{D} \mright) \\
= 
\mathrm{diag} \mleft( \bm{\tau}_{[0:N]} \mright) \bm{R^{0}} \bm{Q} \mleft( \phi, n_{[0:N]}^{D} \mright) 
\mleft\{ \bm{I} + \bm{Q} \mleft( \phi, n_{[0:N]}^{D} \mright) \mright\}^{-1}
, \label{Eq_H_phi_n}
\end{multline}
\begin{equation}
    \bm{Q} \mleft( \phi, n_{[0:N]}^{D} \mright)
= 
\mleft( \bm{\tilde{R}^{D}}\mleft( \phi \mright) \mright)^{-1} \bm{N^{D}}
, \label{Eq_Q_phi_n}
\end{equation}
\begin{equation}
    \bm{N^{D}} =
\begin{bmatrix}
    n_{0}^{D}  & 0      & \cdots & 0 \\
    n_{1}^{D}  & n_{0}^{D}  & \ddots & 0 \\
    \vdots & \vdots & \ddots & 0 \\
    n_{N}^{D}  & n_{N-1}^{D}   & \cdots & n_{0}^{D}
 \end{bmatrix}
, \label{Eq_N_D}
\end{equation}
and $\bm{I}$ denotes the identity matrix with the compatible dimension.
\end{proposition}
\begin{proof} \label{Proof_Prop_FR_ITAE_min_Noisy}
    The noisy fictitious reference $\tilde{r}_{k}^{nD} \mleft( \phi \mright)$ can be represented as 
\begin{align}
    \tilde{r}_{k}^{nD} \mleft( \phi \mright)
&= \mleft\{ C\mleft( z; \phi \mright) \mright\}^{-1} \ast u_{k}^{D} + \mleft( y_{k}^{D} + n_{k}^{D} \mright) \\
&= \tilde{r}_{k}^{D} \mleft( \phi \mright) + n_{k}^{D}
. \label{Eq_R_tilde_Noisy_align}
\end{align}
Substituting \eqref{Eq_R_tilde_Noisy_align} into \eqref{Eq_R_tilde_Noisy} yields that 
\begin{equation}
    \bm{\tilde{R}^{nD}}\mleft( \phi \mright) 
=
\bm{\tilde{R}^{D}}\mleft( \phi \mright) 
+
\bm{{N}^{D}}
. \label{Eq_R_tilde_Noisy_R_tilde_N_D}
\end{equation}
According to the Sherman--Morrison--Woodbury formula (subsection 2.1.4 in \cite{Golub2013}), 
\begin{multline} 
    \mleft( \bm{\tilde{R}^{nD}}\mleft( \phi \mright) \mright)^{-1} 
= \mleft( \bm{\tilde{R}^{D}}\mleft( \phi \mright) + \bm{{N}^{D}} \mright) ^{-1} \\
\shoveleft{=\mleft( \bm{\tilde{R}^{D}}\mleft( \phi \mright) \mright) ^{-1}} \\
- \mleft( \bm{\tilde{R}^{D}}\mleft( \phi \mright) \mright) ^{-1} \bm{{N}^{D}} 
   \mleft\{ \bm{I} + \mleft( \bm{\tilde{R}^{D}}\mleft( \phi \mright) \mright) ^{-1} \bm{{N}^{D}} \mright\}^{-1} 
   \mleft( \bm{\tilde{R}^{D}} \mleft( \phi \mright) \mright) ^{-1} \\
\shoveleft{= \mleft( \bm{\tilde{R}^{D}}\mleft( \phi \mright) \mright) ^{-1} } \\
 - \bm{Q} \mleft( \phi, n_{[0:N]}^{D} \mright)
   \mleft\{ \bm{I} + \bm{Q} \mleft( \phi, n_{[0:N]}^{D} \mright) \mright\}^{-1}
   \mleft( \bm{\tilde{R}^{D}}\mleft( \phi \mright) \mright) ^{-1}
.
 \label{Eq_R_tilde_Noisy_SMW_multline}
\end{multline} 
Here, 
\begin{equation}
    \bm{y^{nD}}_{[0:N]} = \bm{y^{D}}_{[0:N]} + \bm{n^{D}}_{[0:N]}
. \label{Eq_y_D_0_N_Noisy_yD_nD}
\end{equation}
Substituting \eqref{Eq_R_tilde_Noisy_SMW_multline} and \eqref{Eq_y_D_0_N_Noisy_yD_nD} into \eqref{Eq_JD_Noisy} provides the desired conclusion.
\end{proof}

Note that the performance index \eqref{Eq_JD} of FR-ITAE-min based on the clean output data can be represented as 
$J_{ITAE}^{D} \mleft( \phi \mright) = \left\lVert \bm{\ell}_{ITAE}^{D} \mleft( \phi \mright) \right\rVert_{1}$.
Therefore, compared with the noiseless case, the measurement noise on the output data produces the bias $\bm{d}\mleft( \phi, n_{[0:N]}^{D} \mright)$, which depends on both the controller parameter $\phi$ to be tuned and the noise $n_{[0:N]}^{D}$ corrupting the output data. 
Moreover, the existence of the bias implies that, when the output data contains measurement noise, the estimated impulse response $\bm{t}^{n} \mleft( \phi \mright) $ of $T \mleft( z; \phi \mright)$ computed as in \eqref{Eq_t_phi_Ry}, i.e., 
$\bm{t}^{n}\mleft( \phi \mright) 
=
\mleft\{ \bm{\tilde{R}^{nD}} \mleft( \phi \mright) \mright\}^{-1} \bm{y^{nD}}_{[0:N]}$
, contains estimation error.

If the noise on the output data is intensive, the data should be pre-processed by some suitable denoising techniques (e.g., low-pass filtering, total variation denoising \cite{Yahagi2024, Yahagi2021a}, discrete Fourier transform for periodically extended data \cite{Sakai2022}) to attenuate the effect of $\bm{d}\mleft( \phi, n_{[0:N]}^{D} \mright)$ over the search space of $\phi$. Moreover, we will explore a novel reformulation of the optimization problem to enhance robustness to the noise as future work.

\section{Extension of the proposed approach to the minimization of the IAE }  \label{Appendix_IAE}
The IAE, as well as the ITAE, is often used as a performance index for tuning FO controllers \cite{Padula2012, Padula2024}.
In the DT setting, the IAE performance index $J_{IAE} \mleft( \phi \mright)$ is defined as follows:
\begin{equation} \label{Eq_J_IAE_Conventional}
    J_{IAE} \mleft( \phi \mright) \triangleq \sum_{k=0}^{N} \left\lvert y_{k} \mleft( \phi \mright) - r^{0} \right\rvert 
.
\end{equation}
Previous studies have thoroughly discussed the characteristics of the ITAE and IAE criteria (e.g., \cite{Das2011}; see also Section 5.7 of \cite{Dorf2021}).

We can extend the proposed fictitious-reference-based approach to minimize the IAE as follows:
%% Proposition 2 %%
\begin{proposition} \label{Proposition_FR_IAE_min}
    The IAE performance index \eqref{Eq_J_IAE_Conventional} has the following data-driven representation:
     \begin{equation}\label{Eq_J_JD_IAE}
         J_{IAE} \mleft( \phi \mright) = J_{IAE}^{D} \mleft( \phi \mright)
,
     \end{equation}
     \begin{equation} \label{Eq_JD_IAE}
        J_{IAE}^{D} \mleft( \phi \mright)
= 
\left\lVert 
\bm{R^{0}} \mleft( \bm{\tilde{R}^{D}}\mleft( \phi \mright) \mright)^{-1}\bm{y^{D}}_{[0:N]}
-
r^{0} \bm{1}_{[0:N]}
\right\rVert_{1} 
,
     \end{equation}
provided that $\bm{\tilde{R}^{D}}\mleft( \phi \mright)$ is invertible.
\end{proposition}
%%%%%%%%%%%%%%%
\begin{proof}
    It can be proven using the same approach as in Theorem \ref{Theorem_FR_ITAE_min}.
\end{proof}

Note that, similar to the case of the ITAE, \eqref{Eq_JD_IAE} can be evaluated using a single set of the input data $u_{[0:N]}^D$
and the output data $y_{[0:N]}^D$. Therefore, by minimizing $J_{IAE}^{D} \mleft( \phi \mright)$ instead of $J_{IAE} \mleft( \phi \mright)$, we can avoid not only plant but also iterative closed-loop control tests for controller performance assessment.

%% Bibliography %%
\bibliographystyle{ieeetr}
\bibliography{References_ASJC_Ver1_R2.bib}

\begin{thebibliography}{10}

\bibitem{Podlubny1998}
I.~Podlubny, {\em Fractional differential equations}.
\newblock Academic Press, 1998.

\bibitem{Sun2018}
H.~Sun, Y.~Zhang, D.~Baleanu, W.~Chen, and Y.~Chen, ``A new collection of real world applications of fractional calculus in science and engineering,'' {\em Communications in Nonlinear Science and Numerical Simulation}, vol.~64, pp.~213--231, 11 2018.

\bibitem{Cheng2023}
C.~Cheng and H.~Shen, ``Fractional-order dynamics and adaptive dynamic surface control of flexible-joint robots,'' {\em Asian Journal of Control}, vol.~25, pp.~3029--3044, 7 2023.

\bibitem{Xu2025}
Z.~Xu, J.~Wu, and Y.~Wang, ``Fractional order modeling and internal model control method for dielectric elastomer actuator,'' {\em Asian Journal of Control}, vol.~27, pp.~117--127, 1 2025.

\bibitem{Naifar2022}
O.~Naifar, A.~Jmal, and A.~B. Makhlouf, ``Non‐fragile {$H_{\infty}$} observer for lipschitz conformable fractional-order systems,'' {\em Asian Journal of Control}, vol.~24, pp.~2202--2212, 9 2022.

\bibitem{Jmal2020}
A.~Jmal, M.~Elloumi, O.~Naifar, A.~B. Makhlouf, and M.~A. Hammami, ``State estimation for nonlinear conformable fractional-order systems: A healthy operating case and a faulty operating case,'' {\em Asian Journal of Control}, vol.~22, pp.~1870--1879, 9 2020.

\bibitem{Mseddi2024}
A.~Mseddi, K.~Wali, A.~Abid, O.~Naifar, M.~Rhaima, and L.~Mchiri, ``Advanced modeling and control of wind conversion systems based on hybrid generators using fractional order controllers,'' {\em Asian Journal of Control}, vol.~26, pp.~1103--1119, 5 2024.

\bibitem{Wang2023}
S.~Wang, Y.~Sun, X.~Li, B.~Han, and Y.~Luo, ``An optimal robust design method for fractional-order reset controller,'' {\em Asian Journal of Control}, vol.~25, pp.~1086--1101, 3 2023.

\bibitem{BenMakhlouf2023}
A.~B. Makhlouf and O.~Naifar, ``On the barbalat lemma extension for the generalized conformable fractional integrals: Application to adaptive observer design,'' {\em Asian Journal of Control}, vol.~25, pp.~563--569, 1 2023.

\bibitem{Gassara2022}
H.~Gassara, O.~Naifar, A.~B. Makhlouf, and L.~Mchiri, ``Global practical conformable stabilization by output feedback for a class of nonlinear fractional‐order systems,'' {\em Mathematical Problems in Engineering}, vol.~2022, pp.~1--10, 6 2022.

\bibitem{Wang2024}
S.~Wang, B.~Li, P.~Chen, W.~Yu, Y.~Peng, and Y.~Luo, ``A fractional-order active disturbance rejection control for permanent magnet synchronous motor position servo system,'' {\em Asian Journal of Control}, vol.~26, pp.~3137--3147, 11 2024.

\bibitem{Wang2023_ISA}
S.~Wang, H.~Gan, Y.~Luo, X.~Wang, and Z.~Gao, ``Active disturbance rejection control with fractional-order model-aided extended state observer,'' {\em ISA Transactions}, vol.~142, pp.~527--537, 11 2023.

\bibitem{Gnaneshwar2024}
K.~Gnaneshwar, R.~Trivedi, S.~Sharma, and P.~K. Padhy, ``A frequency domain fractional order tilted integral derivative controller design for fractional order time delay processes,'' {\em Asian Journal of Control}, 2024.

\bibitem{Lu2023}
C.~Lu, R.~Tang, Y.~Chen, and C.~Li, ``Robust tilt-integral-derivative controller synthesis for first-order plus time delay and higher-order systems,'' {\em International Journal of Robust and Nonlinear Control}, vol.~33, pp.~1566--1592, 2 2023.

\bibitem{Podlubny1999}
I.~Podlubny, ``Fractional-order systems and {$PI^{\lambda}D^{\mu}$}-controllers,'' {\em IEEE Transactions on Automatic Control}, vol.~44, pp.~208--214, 1 1999.

\bibitem{Chen2022}
P.~Chen, W.~Zheng, Y.~Luo, Y.~Peng, and Y.~Chen, ``Robust three-parameter fractional-order proportional integral derivative controller synthesis for permanent magnet synchronous motor speed servo system,'' {\em Asian Journal of Control}, vol.~24, pp.~3418--3433, 11 2022.

\bibitem{Ghorbani2023}
M.~Ghorbani, A.~Tepljakov, and E.~Petlenkov, ``Stabilizing region of fractional-order proportional integral derivative controllers for interval delayed fractional-order plants,'' {\em Asian Journal of Control}, vol.~25, pp.~1145--1155, 3 2023.

\bibitem{Monje2023}
C.~A. Monje, B.~Deutschmann, J.~Muñoz, C.~Ott, and C.~Balaguer, ``Fractional order control of continuum soft robots: Combining decoupled/reduced-dynamics models and robust fractional order controllers for complex soft robot motions,'' {\em IEEE Control Systems}, vol.~43, pp.~66--99, 6 2023.

\bibitem{Marinangeli2018}
L.~Marinangeli, F.~Alijani, and S.~H. HosseinNia, ``Fractional-order positive position feedback compensator for active vibration control of a smart composite plate,'' {\em Journal of Sound and Vibration}, vol.~412, pp.~1--16, 2018.

\bibitem{Lino2017}
P.~Lino, G.~Maione, S.~Stasi, F.~Padula, and A.~Visioli, ``Synthesis of fractional-order {PI} controllers and fractional-order filters for industrial electrical drives,'' {\em IEEE/CAA Journal of Automatica Sinica}, vol.~4, pp.~58--69, 1 2017.

\bibitem{Chen2006}
Y.~Chen, ``Ubiquitous fractional order controls?,'' {\em IFAC Proceedings Volumes}, vol.~39, pp.~481--492, 1 2006.

\bibitem{Tepljakov2021}
A.~Tepljakov, B.~B. Alagoz, C.~Yeroglu, E.~A. Gonzalez, S.~H. Hosseinnia, E.~Petlenkov, A.~Ates, and M.~Cech, ``Towards industrialization of {FOPID} controllers: A survey on milestones of fractional-order control and pathways for future developments,'' {\em IEEE Access}, vol.~9, pp.~21016--21042, 2021.

\bibitem{Chen2022_ISA}
P.~Chen and Y.~Luo, ``An analytical synthesis of fractional order {$\mathrm{PI}^{\lambda}\mathrm{D}^{\mu}$} controller design,'' {\em ISA Transactions}, vol.~131, pp.~124--136, 12 2022.

\bibitem{Wu2024}
Z.~Wu, J.~Viola, Y.~Luo, Y.~Chen, and D.~Li, ``Robust fractional-order [proportional integral derivative] controller design with specification constraints: more flat phase idea,'' {\em International Journal of Control}, vol.~97, pp.~111--129, 1 2024.

\bibitem{Yumuk2022}
E.~Yumuk, M.~G{\"{u}}zelkaya, and {\.{I}}.~Eksin, ``A robust fractional-order controller design with gain and phase margin specifications based on delayed bode's ideal transfer function,'' {\em Journal of the Franklin Institute}, vol.~359, pp.~5341--5353, 7 2022.

\bibitem{Shankaran2022}
V.~P. Shankaran, S.~I. Azid, and U.~Mehta, ``Fractional-order {PI} plus {D} controller for second-order integrating plants: Stabilization and tuning method,'' {\em ISA Transactions}, vol.~129, pp.~592--604, 10 2022.

\bibitem{Das2011}
S.~Das, S.~Saha, S.~Das, and A.~Gupta, ``On the selection of tuning methodology of {FOPID} controllers for the control of higher order processes,'' {\em ISA Transactions}, vol.~50, pp.~376--388, 7 2011.

\bibitem{Bingul2018}
Z.~Bingul and O.~Karahan, ``Comparison of {PID} and {FOPID} controllers tuned by {PSO} and {ABC} algorithms for unstable and integrating systems with time delay,'' {\em Optimal Control Applications and Methods}, vol.~39, pp.~1431--1450, 7 2018.

\bibitem{Mondal2020}
R.~Mondal, A.~Chakraborty, J.~Dey, and S.~Halder, ``Optimal fractional order {$\mathrm{PI}^{\lambda}\mathrm{D}^{\mu}$} controller for stabilization of cart-inverted pendulum system: Experimental results,'' {\em Asian Journal of Control}, vol.~22, pp.~1345--1359, 5 2020.

\bibitem{Faraj2023}
M.~A. Faraj, B.~Maalej, N.~Derbel, and O.~Naifar, ``Adaptive fractional‐order super‐twisting sliding mode controller for lower limb rehabilitation exoskeleton in constraint circumstances based on the grey wolf optimization algorithm,'' {\em Mathematical Problems in Engineering}, vol.~2023, 1 2023.

\bibitem{Paliwal2022}
N.~Paliwal, L.~Srivastava, and M.~Pandit, ``Rao algorithm based optimal multi-term {FOPID} controller for automatic voltage regulator system,'' {\em Optimal Control Applications and Methods}, vol.~43, pp.~1707--1734, 11 2022.

\bibitem{Veerendar2023}
T.~Veerendar, D.~Kumar, and V.~Sreeram, ``Fractional-order {PID} and internal model control-based dual-loop load frequency control using teaching--learning optimization,'' {\em Asian Journal of Control}, vol.~25, pp.~2482--2497, 7 2023.

\bibitem{Hekimoglu2019}
B.~Hekimoglu, ``Optimal tuning of fractional order {PID} controller for {DC} motor speed control via chaotic atom search optimization algorithm,'' {\em IEEE Access}, vol.~7, pp.~38100--38114, 2019.

\bibitem{Izci2023}
D.~Izci and S.~Ekinci, ``A novel-enhanced metaheuristic algorithm for {FOPID}-controlled and bode's ideal transfer function--based buck converter system,'' {\em Transactions of the Institute of Measurement and Control}, vol.~45, pp.~1854--1872, 6 2023.

\bibitem{Deniz2020}
F.~N. Deniz, B.~B. Alagoz, N.~Tan, and M.~Koseoglu, ``Revisiting four approximation methods for fractional order transfer function implementations: Stability preservation, time and frequency response matching analyses,'' {\em Annual Reviews in Control}, vol.~49, pp.~239--257, 2020.

\bibitem{Li2011}
Y.~Li, H.~Sheng, and Y.~Chen, ``Analytical impulse response of a fractional second order filter and its impulse response invariant discretization,'' {\em Signal Processing}, vol.~91, pp.~498--507, 3 2011.

\bibitem{Tavazoei2010}
M.~S. Tavazoei, ``Notes on integral performance indices in fractional-order control systems,'' {\em Journal of Process Control}, vol.~20, pp.~285--291, 2010.

\bibitem{Dorf2021}
R.~Dorf and R.~Bishop, {\em Modern Control Systems}.
\newblock Pearson, 14th~ed., 2021.

\bibitem{Safonov1997}
M.~G. Safonov and T.~C. Tsao, ``The unfalsified control concept and learning,'' {\em IEEE Transactions on Automatic Control}, vol.~42, pp.~843--847, 1997.

\bibitem{Kaneko2013}
O.~Kaneko, ``Data-driven controller tuning: {FRIT} approach,'' {\em IFAC Proceedings Volumes}, vol.~46, pp.~326--336, 2013.

\bibitem{Baldi2011}
S.~Baldi, G.~Battistelli, D.~Mari, E.~Mosca, and P.~Tesi, ``Multi-model adaptive switching control with fine controller tuning,'' {\em IFAC Proceedings Volumes}, vol.~44, pp.~374--379, 1 2011.

\bibitem{Yahagi2024}
S.~Yahagi and I.~Kajiwara, ``One-shot data-driven design for feedback controller and reference model with {BIBO} stability,'' {\em IEEE Access}, vol.~12, pp.~147882--147893, 2024.

\bibitem{Yonezawa2024}
A.~Yonezawa, H.~Yonezawa, S.~Yahagi, and I.~Kajiwara, ``Practical one-shot data-driven design of fractional-order {PID} controller: Fictitious reference signal approach,'' {\em ISA Transactions}, vol.~152, pp.~208--216, 9 2024.

\bibitem{Padula2014}
F.~Padula and A.~Visioli, ``Inversion-based feedforward and reference signal design for fractional constrained control systems,'' {\em Automatica}, vol.~50, pp.~2169--2178, 2014.

\bibitem{Yahagi2021a}
S.~Yahagi and I.~Kajiwara, ``Direct tuning of the data-driven controller considering closed-loop stability based on a fictitious reference signal,'' {\em Measurement and Control}, vol.~54, pp.~1026--1042, 2021.

\bibitem{Sakai2022}
T.~Sakai, S.~Yahagi, and I.~Kajiwara, ``Two-degree-of-freedom controller design based on a data-driven estimation approach,'' {\em IEEE Access}, vol.~10, pp.~120475--120491, 2022.

\bibitem{Wang2000}
Y.-G. Wang and H.-H. Shao, ``{PID} auto-tuner based on sensitivity specification,'' {\em Chemical Engineering Research and Design}, vol.~78, pp.~312--316, 3 2000.

\bibitem{Das2018}
S.~Das, K.~Halder, and A.~Gupta, ``Performance analysis of robust stable {PID} controllers using dominant pole placement for {SOPTD} process models,'' {\em Knowledge-Based Systems}, vol.~146, pp.~12--43, 4 2018.

\bibitem{Dastjerdi2019}
A.~A. Dastjerdi, B.~M. Vinagre, Y.~Chen, and S.~H. HosseinNia, ``Linear fractional order controllers; a survey in the frequency domain,'' {\em Annual Reviews in Control}, vol.~47, pp.~51--70, 2019.

\bibitem{Oustaloup2000}
A.~Oustaloup, F.~Levron, B.~Mathieu, and F.~Nanot, ``Frequency-band complex noninteger differentiator: characterization and synthesis,'' {\em IEEE Transactions on Circuits and Systems I: Fundamental Theory and Applications}, vol.~47, pp.~25--39, 2000.

\bibitem{Tepljakov2011}
A.~Tepljakov, E.~Petlenkov, and J.~Belikov, ``{FOMCON}: Fractional-order modeling and control toolbox for matlab,'' in {\em Proceedings of the 18th International Conference - Mixed Design of Integrated Circuits and Systems, {MIXDES} 2011}, pp.~684--689, IEEE, 2011.

\bibitem{Golub2013}
G.~H. Golub and C.~V. Loan, {\em Matrix Computations}.
\newblock Johns Hopkins University Press, 4th~ed., 2013.

\bibitem{Padula2012}
F.~Padula and A.~Visioli, ``Optimal tuning rules for proportional-integral-derivative and fractional-order proportional-integral-derivative controllers for integral and unstable processes,'' {\em IET Control Theory \& Applications}, vol.~6, pp.~776--786, 2012.

\bibitem{Padula2024}
F.~Padula and A.~Visioli, ``A general design methodology for fractional cascade control systems,'' {\em IEEE Access}, vol.~12, pp.~11451--11457, 2024.

\end{thebibliography}
%%%%%%%%%%%%%%%%%%

\end{document}